\documentclass{siamltex}

\usepackage{amsfonts, amsmath, amssymb, psfrag, graphicx, bbm, mathrsfs}

\usepackage{soul}

\usepackage[colorlinks=true]{hyperref}
\hypersetup{urlcolor=blue, citecolor=blue, linkcolor=blue}

\newtheorem{remark}[theorem]{Remark}

\newtheorem{theoremold}{Theorem}

\def\fin{\ifmmode{\Large$\diamond$}\else{\unskip\nobreak\hfil
    \penalty50\hskip1em\null\nobreak\hfil{\Large$\diamond$}
    \parfillskip=0pt\finalhyphendemerits=0\endgraf}\fi}

\def\be#1#2\ee{\begin{equation}\label{eq:#1}#2\end{equation}}
\def\req#1{{\rm(\ref{eq:#1})}}
\def\bdm  {\begin{displaymath}}
  \def\edm  {\end{displaymath}}
\def\bdmal{\begin{displaymath}\begin{aligned}}
    \def\edmal{\end{aligned}\end{displaymath}}

\def\thsp{\hspace*{0.1ex}}

\mathchardef\PhiG="0108

\renewcommand{\L}{{\mathscr L}}
\newcommand{\N}{{\mathord{\mathbb N}}}
\newcommand{\R}{{\mathord{\mathbb R}}}

\newcommand{\E}{{\mathord{\mathbb E}}}
\renewcommand{\S}{{\mathord{\mathbb S}}}
\renewcommand{\P}{{\mathord{\mathbb P}}}

\newcommand{\xx}{\textit{\textbf{x}}}
\newcommand{\yy}{\textit{\textbf{y}}}

\newcommand{\B}{{\cal B}}
\newcommand{\BV}{\B_\V}

\renewcommand{\L}{{\mathscr L}}

\newcommand{\norm}[1]{\|#1\|}

\newcommand{\scalp}[1]{\langle\,#1\,\rangle}

\newcommand{\rmd}{\,\mathrm{d}}

\newcommand{\dr}{\rmd r}

\newcommand{\dx}{\rmd x}

\newcommand{\dyy}{\rmd \yy}

\newcommand{\eps}{\varepsilon}
\newcommand{\pcf}{\rho^{(2)}}

\def\dom#1{{\mathcal D}(#1)}

\def\rcp{{\rho_{\textnormal{cp}}}}
\def\req#1{{\rm(\ref{eq:#1})}}

\makeatletter
\newcommand{\dupdots}{\mathinner{\mkern1mu\raise\p@
    \vbox{\kern7\p@\hbox{.}}\mkern2mu
    \raise4\p@\hbox{.}\mkern2mu\raise7\p@\hbox{.}\mkern1mu}}
\makeatother
\makeatletter
\newcount\c@MaxMatrixCols \c@MaxMatrixCols=10

\makeatother

 {\endMakeFramed}

\newcommand{\U}{{{\mathscr U}}}
\newcommand{\V}{{{\mathscr V}}}
\newcommand{\X}{{\mathscr X}}

\newcommand{\bfrho}{{\boldsymbol{\rho}}}

\newcommand{\trho}{\widetilde\rho}
\newcommand{\vt}{\widetilde v}
\newcommand{\Vt}{\widetilde V}

\newcommand{\dP}{\,\rmd\P}

\newcommand{\cbeta}{c_\beta}

\newcommand{\bfe}{{\boldsymbol{e_1}}}
\newcommand{\Srel}{S_{\rm rel}}
\newcommand{\SrelLambda}{S_{{\rm rel},\Lambda}}

\makeatletter
\renewcommand\@biblabel[1]{#1.}
\makeatother

\title{
A variational framework for the inverse Henderson problem of statistical
mechanics\thanks{This article belongs to the themed collection: 
    Mathematical Physics and Numerical Simulation of Many-Particle Systems; 
    V. Bach and L. Delle Site (eds.).
    The research leading to this work
    has been done within the 
    Collaborative Research Center TRR~146; corresponding funding 
    by the DFG is gratefully acknowledged.}}
\author{Fabio Frommer\thanks{Institut f\"ur Mathematik, Johannes
    Gutenberg-Universit\"at Mainz, 55099 Mainz, Germany
    ({\tt fabiofrommer@uni-mainz.de})} \and 
    Martin Hanke\thanks{Institut f\"ur Mathematik, Johannes
    Gutenberg-Universit\"at Mainz, 55099 Mainz, Germany
    ({\tt hanke@math.uni-mainz.de})}}

\begin{document}
\sloppy
\maketitle

\begin{abstract}
The inverse Henderson problem refers to the determination of the pair
potential which specifies the interactions in an ensemble of 
classical particles in continuous space, given the density and 
the equilibrium pair correlation function of these particles as data.
For a canonical ensemble in a bounded domain
it has been observed that this pair potential minimizes a corresponding
convex relative entropy functional, and that
the Newton iteration for minimizing this functional coincides with the 
so-called inverse Monte Carlo (IMC) {iterative} scheme. 
In this paper we show that in the thermodynamic limit analogous connections 
exist between the specific relative entropy introduced by Georgii and Zessin 
and a proper formulation of the IMC iteration in the full space.
This provides a rigorous variational framework for the 
inverse Henderson problem, valid within a large class of pair potentials,
including, for example, Lennard-Jones type potentials. 

It is further shown that the pressure is 
strictly convex as a function of the pair potential and the chemical potential,
and that the specific relative entropy at fixed density is a strictly convex
function of the pair potential.
At a given reference potential and a corresponding density in the gas phase
we determine the gradient and the 
Hessian of the specific relative entropy,
and we prove that the Hessian extends to a symmetric positive semidefinite 
quadratic functional in the space of square integrable perturbations 
of this potential.
\end{abstract}


\begin{AMS}
{\sc 82B21, 82B80, 60G55}
\end{AMS}

\hspace*{-0.7em}
{\footnotesize \textbf{Last modified.} \today}

%

\section{Introduction}
\label{Sec:Introduction}
Numerical simulation has established itself as an independent and
indispensable branch of research in the natural sciences, on equal footing
with theory and experiment. To be truly useful, numerical approaches have 
to face and master the multiscale nature which is ubiquitous in almost
all relevant applications. 
This is particularly true for soft matter, where spatial
scales may bridge from the electron scale up to the millimetre scale
of biomaterials or polymers. Concerning examples we refer to 
the excellent survey by Noid~\cite{Noid13a},
the collected volume edited by Monticelli and Salonen~\cite{MoSa13},
or the recent special issue~\cite{Schm22} 
of \emph{Journal of Physics: Condensed Matter}.

An important technique to advance numerical algorithms to the particular needs
of multiscale applications consists in \emph{coarse-graining},
cf., e.g., Noid~\cite{Noid13b}, or Peter and Kremer~\cite{PeKr09}:
small-scale features are discarded on the coarser scale by replacing detailed 
descriptions of molecules or matter by artificial \emph{beads} 
of a certain shape. The simulation then focusses only on these beads and their
interaction with the other constituents of the system. 
When and where necessary, fine details can be reinserted back
into the simulation for better accuracy as, e.g., in the AdResS scheme
(\cite{PDK05,FJK12}).

Of course, to evaluate the equations of motion for the coarse-grained model
it is necessary to derive the prevailing effective forces on these beads.
Concerning the transition from an atomic microscale to a molecular
macroscale in thermodynamic equilibrium, for example, there are essentially
two ways to settle this problem. On the one hand 
one can employ an \emph{ab initio} bottom-up approach and evaluate and assemble 
the resulting forces from the eliminated details of the fine-grained
description, cf.~Ercolessi and Adams~\cite{ErAd94},
or run a fine-grained simulation, compute the
corresponding forces, and somehow approximate them on the
coarse level as suggested, e.g., by Izvekov and Voth~\cite{IzVo05}, 
and Wang et al~\cite{WNLV09}.
The other alternative is to follow a top-down strategy and use the given
structural information about the location of the beads on the coarse scale
to formulate an \emph{inverse problem}: Which are the appropriate
forces or interactions on the coarse level that define ensembles with the
same structural properties\thsp?

In this work we treat one of the simplest incarnations 
of the latter approach. Let us presume that the probabilities of
the snapshots of the coarse-grained bead ensembles are in good agreement 
with a model which only uses additive pairwise translation 
invariant interactions of the beads.
Such interactions can be formulated in terms of a scalar pair potential 
which only depends on the relative position of the respective pair of beads. 
The pair correlation function, which measures the empirical likelihood 
to observe two beads at a given relative position, 
appears to be an adequate piece of data to be used for finding 
the corresponding potential, because both the data and the unknown 
consist of a scalar function of the space variable in that case. 
In fact, in a celebrated paper, Henderson~\cite{Hend74} 
argued that the pair potential is uniquely specified 
this way, i.e., under given coditions of temperature and density,
no two different pair potentials can give rise to the same piece
of empirical data. 
Finding a pair potential that reproduces the given pair correlation function
is therefore sometimes called the \emph{inverse Henderson problem}.

The numerical solution of the inverse Henderson problem is demanding
because of the lack of a mathematical formula for computing 
the pair correlation function for a given pair potential,
or vice versa. In the old days people have developed approximate identities
like the Percus-Yevick or hypernetted chain integral equations 
for this purpose, cf.~Hansen and McDonald~\cite{HaMcD13}, 
or have used parametric ansatz functions like, e.g., Lennard-Jones potentials, 
and optimized the corresponding parameters numerically. 
Today, the state of the art is to use non-parametric potentials and
employ iterative schemes which, in each iteration, simulate the
equilibrium ensemble for the current guess of the pair potential, and
use the associated data fit to somehow generate a new guess.
Well-known examples are the iterative Boltzmann inversion (IBI), 
cf.~Soper~\cite{Sope96} and Reith, P\"utz, and M\"uller-Plathe~\cite{RPM03},
and the inverse Monte-Carlo method (IMC) by 
Lyubartsev and Laaksonen~\cite{LyLa95}. 
We recommend the valuable reviews by Toth~\cite{Toth07} and 
R\"uhle et al~\cite{RJLKA09} for a comparison of these and further methods.

We emphasize that the setting of the inverse Henderson problem is generally 
accepted to be far too simplistic to capture all the relevant 
features of a real system, cf., e.g., \cite{LLV10,Noid13b,THT17}, 
mostly because multibody interactions are neglected. 
In particular, thermodynamic properties
of the coarse-grained model may differ from the real system, 
especially at other temperatures or densities.
 
But its simplicity offers a great opportunity for a 
mathematical analysis, which in turn may lead to a better understanding of 
other, more flexible coarse-graining techniques that are routinely being 
employed in practice. Still, only few rigorous mathematical results 
have yet been obtained, 
e.g., in \cite{JTC83,Koral07,Kuna7,Navr16,Hank18c,FHJ19},
the reason being, again, the lack of explicit formulae to attack the problem.

The aim of this paper is to point out and advocate an alternative access point
for theoretical investigations, which goes back to a nice observation by
Shell~\cite{Shel08} from within the chemical physics community: 
He argues that the Henderson potential
minimizes the (information theoretic) relative entropy
\be{Shell}
   \Srel \,=\, 
   \sum_\gamma {\cal P}_*(\gamma)\,
              \log\frac{{\cal P}_*(\gamma)}{{\cal P}(\gamma)}\,,
\ee
where -- in his words --
the summation is over all possible (coarse grained)
ensemble configurations $\gamma$,
$\mathcal{P}_*(\gamma)$ denotes the target (or observed) probability of 
$\gamma$, and $\mathcal{P}(\gamma)$ is the corresponding probability of 
a model with a given pair potential. 
(Compare, for example, Georgii~\cite{Geo13} for background on the concept of
relative entropy in stochastics.)
Subsequently, Murtola, Karttunen, and Vattulainen~\cite{MKV09}
pointed out that the functional~\req{Shell} is convex, and that
the Newton iteration for minimizing the relative entropy coincides
with the aforementioned IMC iteration
(see also Rosenberger et al~\cite{RSSV19}).

Like Henderson's paper~\cite{Hend74}, the results in \cite{MKV09,RSSV19,Shel08}
lack mathematical rigor, because their arguments are restricted to
bounded domains 
-- whereas an unambiguous definition of a ``translation invariant ensemble''
is only possible in the full space.
Concerning the Henderson theorem this shortcoming has recently been fixed 
in \cite{FHJ19}, building on fundamental work by Ruelle and by Georgii.
Here we focus on the proposal by Shell and his colleagues, provide a
rigorous justification of their results, and elaborate further on them.

To be specific we outline in Section~\ref{Sec:thermodynamic_limit} 
that the (appropriately formulated) relative entropy~\req{Shell} divided by 
the volume of the bounded domain converges in the thermodynamic limit
to the specific relative entropy, first introduced for continuum
systems by Gallavotti and Miracle-Sole~\cite{Gall68} in the case of 
hard-core interactions, and further investigated
by Georgii and Zessin in \cite{HGEO93,HGEO94,HGEO95} for 
general additive pair interactions. 
Under mild assumptions on the model and target ensembles 
to be utilized we prove that this specific relative entropy is a 
strictly convex functional on the corresponding set of pair potentials
(which include Lennard-Jones type 
and hard-core potentials),
and that its (unique) minimizer is the particular potential
which solves the inverse Henderson problem
(see Sections~\ref{Sec:convex} and \ref{Sec:rhofixed}).

From Section~\ref{Sec:derivatives} onwards we restrict ourselves to
low densities (the ``gas phase''), where the specific relative entropy is a
differentiable function of the pair potential. We calculate its Hessian,
and in Section~\ref{Sec:IMC} we verify that the Newton iteration for
minimizing the specific relative entropy does indeed coincide with the 
IMC iteration formulated in the thermodynamic limit. 
In Section~\ref{Sec:L2} we investigate the Hessian in more detail and show
that it can be represented by a selfadjoint positive semidefinite operator
in $L^2$. The mapping properties of this operator can be analyzed 
somewhat further
for the particular class of Lennard-Jones type pair potentials; we conclude
with a corresponding result in Section~\ref{Sec:LJtype}.

We hope that this variational framework for the inverse Henderson problem
opens a possibility to discuss the convergence of the IMC iteration,
or to come up with measures to stabilize or regularize
this popular iterative scheme.

\section{The relative entropy in the thermodynamic limit}
\label{Sec:thermodynamic_limit}
To reformulate Shell's approach within a rigorous mathematical framework 
we start with the assumption that the target ensemble is given by 
a translation invariant probability measure $\P_*$ on the 
configuration space
\[
   \Gamma \,=\, \{\,\gamma\subset\R^d\,:\, \triangle\subset\R^d \ \text{bounded}
   \,\Rightarrow\, \#(\gamma\cap\triangle) < \infty\,\}
\]
with density $\rho_*$ and finite locally second moments, 
compare Georgii~\cite{HGEO94}. 
For the model ensemble we restrict ourselves to ensembles
of classical particles with additive pairwise interactions defined by 
a measurable even pair potential 
$u\colon\R^d\to\R\cup\{+\infty\}$.
Concerning the latter we assume that there exists $r_0>0$ and 
decreasing positive functions $\varphi:(0,r_0)\to \R^+_0$ and 
$\psi: [0,\infty)\to \R^+$ with
\be{varphi,psi}
   \int_{0}^{r_0} r^{d-1}\varphi(r)\dr = +\infty
   \qquad \text{and} \qquad
   \int_{0}^\infty r^{d-1}\psi(r)\dr <\infty\,,
\ee
such that
\be{potentials}
\begin{aligned}
   u(x) &\,\geq\, \varphi(|x|) \qquad \textrm{ for }\ 0 < |x| < r_0\,, \\[1ex]
   |u(x)| &\,\leq\, \psi(|x|) \qquad \textrm{ for }\ |x| \geq r_0\,,
\end{aligned}
\ee
holds true almost everywhere; for our convenience we take $\psi$
to be a bounded and decreasing function defined for \emph{all} $r\geq 0$.
We denote by $\U_0$ the subset of the above pair potentials, 
which also belong to $L^\infty_{\rm loc}(\R^3\setminus\{0\})$,
and define $\U$ to be the union of $\U_0$ with the hard-core potentials, 
which satisfy \req{potentials} with $\psi$ as above and $\varphi$ replaced
by $+\infty$; 
accordingly, $r_0$ is taken to be the hard-core radius in this case.
Technically we do not distinguish between potentials (and functions in general)
which differ on sets of Lebesgue measure zero.

\begin{remark}
\label{Rem:convex}
\rm
The set $\U$ is convex. To see this
let $u_i\in\U$, $i=1,2$, be such that \req{varphi,psi} and \req{potentials} 
hold for $r_{0,i}>0$ and decreasing functions $\varphi_i$ and $\psi_i$ 
satisfying \req{varphi,psi}, respectively. 
Without loss of generality we may assume that $r_{0,1}\leq r_{0,2}$.
Let $u=tu_1+(1-t)u_2$ for some fixed $t\in(0,1)$.
We distinguish two cases. 
If $u_2$ is a hard-core potential, then $u$ is also a hard-core potential
with hard-core radius $r_{0,2}$. In this case we can choose $r_0=r_{0,2}$,
\[\begin{aligned}
   \varphi(r) &\,=\, +\infty & \quad &
   \textrm{ for }\ 0 < r < r_{0,2}\,,\\[1ex]
   \psi(r) &\,=\, t\psi_1(r)+(1-t)\psi_2(r)  & \quad & 
   \textrm{ for }\ r \geq 0\,,
\end{aligned}
\]
to achieve the assumption \req{potentials} for $u$.
On the other hand, if $u_2$ is no hard-core potential, then $u_2$ is bounded
on $[r_{0,1},r_{0,2})$, and there exists $c\geq 1$, such that
\[
   |u_2(x)| \,\leq\, c\,\psi_2(|x|) \qquad 
   \textrm{ for }\ r_{0,1} \leq |x| < r_{0,2}\,.
\]
It follows that $u$ satisfies the inequalities~\req{potentials} 
with $r_0=r_{0,1}$ and
\[
\begin{aligned}
   \varphi(r) &\,=\, t\varphi_1(r)+(1-t)\varphi_2(r) & \quad & 
   \textrm{ for }\ 0 < r < r_{0,1}\,,\\[1ex]
   \psi(r) &\,=\, t\psi_1(r)+(1-t)c\,\psi_2(r) & \quad & 
   \textrm{ for }\ r\geq 0\,.
\end{aligned}
\]
In either case we have verified that $u\in\U$, hence $\U$ is convex.
\fin
\end{remark}

Let $\Lambda=[-\ell,\ell]^d$ be a bounded box in $\R^d$.
In analogy to Shell, who considered the relative entropy framework for
canonical ensembles in $\Lambda$, we take the
restriction $\P_{*,\Lambda}$ of $\P_*$ as target, and the
grand canonical ensemble in $\Lambda$ with chemical potential $\mu\in\R$ and
pair potential $u\in\U$ as model.
For a configuration $\gamma\in\Gamma$ of $N=N(\gamma)$ particles located at
$\{x_1,\dots,x_N\}\subset\Lambda$ the probability density of the model 
is thus given by
\[
   {\cal P}(\gamma)
   \,=\, \frac{1}{\ \Xi_\Lambda}\,e^{\beta\mu N(\gamma)}e^{-\beta U(\gamma)}\,,
\]
where $\beta>0$ is the inverse temperature, 
\be{U}
   U(\gamma) \,=\, U(x_1,\dots,x_N) 
   \,=\!\! \sum_{1\leq i<j\leq N}\!\! u(x_i-x_j)
\ee
is the interaction energy, and
\[
   \Xi_\Lambda \,=\, 
   \sum_{N=0}^\infty \frac{e^{\beta\mu N}}{N!}
   \int_{\Lambda^N} e^{-\beta U(x_1,\dots,x_N)}\dx_1\cdots\dx_N   
\]
is the corresponding grand canonical partition function. 

Note that most of the quantities we deal with depend on $\beta$;
however, since we will keep the temperature fixed throughout this paper,
we refrain from making this dependency explicit in our notation.

Let $S(\P_{*,\Lambda})$ denote the entropy associated with $\P_{*,\Lambda}$.
In the spirit of \req{Shell} we then specify the relative entropy
\begin{align*}
   \SrelLambda
   &\,=\, S(\P_{*,\Lambda})
    + \int_{\Gamma_\Lambda} \bigl(
      \log{\Xi_\Lambda}
      -\beta\mu N(\gamma)+\beta\thsp U(\gamma)\bigr)
      \dP_{*,\Lambda}(\gamma)\\[1ex]
   &\,=\, S(\P_{*,\Lambda})
    + \log\Xi_\Lambda 
    - \beta\mu \,\E_{*,\Lambda}[N] \,+\, \beta\,\E_{*,\Lambda}[U]\,,
\end{align*}
where $\E_{*,\Lambda}[\,\cdot\,]$ denotes expectation with respect
to $\P_{*,\Lambda}$. Take note that the entropy and the expected 
interaction energy may be $+\infty$. Dividing by $\beta|\Lambda|$ we arrive at
\be{Shell2}
   \frac{1}{\beta|\Lambda|}\,\SrelLambda 
   \,=\, \frac{1}{\beta|\Lambda|} \,S(\P_{*,\Lambda})
    + \frac{1}{\beta |\Lambda|}\log\Xi_\Lambda -  \mu \rho_* 
    + \,\frac{1}{|\Lambda|}\,\E_{*,\Lambda} [U]\,.
\ee

Now we want to derive the analog of this identity for the thermodynamic limit
$\ell\to\infty$. 
Concerning this limit it is known that there is a sequence $(\ell_k)_k$
with $\ell_k\to\infty$, such that the probability densities associated with 
the corresponding grand canonical ensembles converge in a local topology 
to a translation invariant tempered Gibbs measure on $\Gamma$,
cf.~Ruelle~\cite{Ruel70}, for brevity called $(\mu,u)$-Gibbs measure in the 
sequel.
In general different
sequences may lead to different $(\mu,u$)-Gibbs measures, e.g., 
when the system exhibits 
a phase transition. For any such sequence, however, the limit
\be{pressure}
   p(\mu,u) \,=\, 
   \lim_{\ell\to\infty}\,\frac{1}{\beta|\Lambda|}\,\log\Xi_\Lambda 
\ee
is always the same well-defined and nonnegative finite number, 
namely the pressure of the ensemble in the thermodynamic limit, 
cf.~Ruelle~\cite{Ruel69}. Therefore, passing in \req{Shell2}
to the thermodynamic limit $\ell\to\infty$, 
we can (uniquely) define the relative entropy
$\Srel(\mu,u)$ of all these $(\mu,u)$-Gibbs measures 
with respect to the target model $\P_*$ as
\be{Georgii}
   \frac{1}{\beta}\,\Srel(\mu,u)
   \,=\, \frac{1}{\beta}\,S_* \,+\, p(\mu,u) \,-\, \mu \rho_* \,+\, E(u,\P_*)\,,
\ee
where the individual terms on the right-hand side of \req{Georgii} 
are the corresponding
limits of the respective terms in \req{Shell2}: 
$S_*=S(\P_*)$ is the specific entropy 
and $E(u,\P_*)$ is the specific interaction energy (with respect to the
potential $u$) of the target ensemble, both of which have been
shown to be well-defined in $\R\cup\{+\infty\}$, cf.~\cite{Geo11,HGEO93}.
If $\P_*$ satisfies a Ruelle condition 
(compare \req{Ruellebound} in the appendix) 
then the specific entropy is finite;
this is the case, e.g., when $\P_*$ is a $(\mu_*,u_*)$-Gibbs measure for 
some $u_*\in\U$ and $\mu_*\in\R$, compare \cite[Corollary~5.3]{Ruel70}.
For this latter particular case it has further been shown in \cite{FHJ19} that
\be{EuP}
   E(u,\P_*) \,=\, \frac{1}{2} \int_{\R^d} u(x)\rho_*^{(2)}(x)\dx\,,
\ee
where $\rho_*^{(2)}(x)$ is the pair correlation function 
associated with $\P_*$.
Here, and throughout this paper, we deliberately use the short-hand
notation $\pcf(x)$ instead of $\pcf(x,0)$ for the pair correlation function of
a translation invariant point process.

In \cite{HGEO95} Georgii studied the relative entropy $\Srel(\mu,u)$
in detail and established the following \emph{Gibbs variational principle}:

\begin{theoremold}[Gibbs variational principle]
\label{Thm:Georgii}
Let $\P_*$ be a translation invariant probability measure on $\Gamma$ with
finite locally second moments. Then the relative entropy \req{Georgii} 
is nonnegative real or $+\infty$ for every $u\in\U$ and $\mu\in\R$. 
There holds $\Srel(\mu,u)=0$, if and only if
$\P_*$ is a $(\mu,u)$-Gibbs measure.
\end{theoremold}

We have utilized this result in \cite{FHJ19} to prove a rigorous version 
of the Henderson theorem:

\begin{theoremold}[Henderson theorem]
\label{Thm:Henderson}
Let $u_1,u_2\in\U$ and $\mu_1,\mu_2\in\R$.
If $\P_1$ and $\P_2$ are $(\mu_1,u_1)$- and $(\mu_2,u_2)$-Gibbs measures, 
respectively, 
which share the same density and the same pair correlation function, 
then $u_1=u_2$ and $\mu_1=\mu_2$.
\end{theoremold}


Combining Georgii's Gibbs variational principle and the Henderson theorem
we can formulate an alternative version of the Gibbs variational principle.
This version is the one that we will mostly use below.

\begin{theorem}[Gibbs variational principle, alternative form]
\label{Thm:gvp}
If the target $\P_*$ is a 
$(\mu_*,u_*)$-Gibbs measure for some $u_*\in\U$ and $\mu_*\in\R$,
then the relative entropy $\Srel(\mu,u)$ becomes minimal, if and only if 
$u=u_*$ and $\mu=\mu_*$.
\end{theorem}

\begin{proof}
According to Theorem~\ref{Thm:Georgii} it remains to investigate the case
when $\Srel(\mu,u)=0$ for some $u\in\U$ and $\mu\in\R$.
The Gibbs variational principle states that $\P_*$ then is
a $(\mu_*,u_*)$- and $(\mu,u)$-Gibbs measure
at the same time. In particular this means that these two Gibbs measures share
the same density and pair correlation function. The assertion thus
follows from the Henderson theorem.
\end{proof}

Georgii investigated the relative entropy for fixed interaction and chemical 
potentials $u$ and $\mu$, and varied the target model $\P_*$.
In connection with the inverse Henderson problem our interest is sort of dual 
to this: we assume that $\P_*$ is fixed, and consider the relative entropy 
as a function of $\mu$ and $u$. 

\section{Strict convexity of the pressure and the relative entropy}
\label{Sec:convex}
It is well-known that the pressure is a convex function of
the chemical potential, cf.~\cite[Theorem~3.4.6]{Ruel69}. 
This convexity is strict, whenever a Gibbs variational principle is valid,
cf., e.g., Hughes~\cite[Section 4.3]{Hugh21}.
In the sequel we show that under our assumptions the pressure is also strictly
convex in $\mu$ \emph{and} $u$.

\begin{theorem}
\label{Thm:convexpressure}
The pressure $p=p(\mu,u)$ of \req{pressure} is a strictly convex function
of $\mu\in\R$ and $u\in\U$. Moreover, for $u\in\U_0$ there holds
\be{Thm:convexpressure}
   \frac{p(\mu,u)}{\mu} \,\to\, +\infty \qquad \text{as $\mu\to+\infty$}\,.
\ee
\end{theorem}

\begin{proof}
To begin with we first observe that $\R\times\U$ is convex by virtue of 
Remark~\ref{Rem:convex}.

Now let $\mu_1,\mu_2\in\R$ and $u_1,u_2\in\U$ be arbitrarily chosen, 
not both being equal at the same time, and define
\[
   \mu \,=\, t\mu_1 + (1-t)\mu_2 \qquad \text{and} \qquad
   u \,=\, tu_1 + (1-t)u_2
\]
for some fixed $t\in(0,1)$. When choosing for $\P_*$ a corresponding 
$(\mu,u)$-Gibbs measure,
then it follows from \req{Georgii} 
and the Gibbs variational principle of Theorem~\ref{Thm:gvp} that
\be{convexpressure-help}
            p(\mu,u) 
            \,=\, -\frac{1}{\beta}\,S_*
                  \,+\,\mu\rho_*
                  \,-\, E(u,\P_*)\,.
\ee
In particular, the specific entropy $S_*$ is finite because $\P_*$ is a 
Gibbs measure, and hence, so is the specific interaction energy. 
Likewise we obtain
\begin{align*}
            p(\mu_1,u_1) 
            &\,>\, -\frac{1}{\beta}\,S_*
                  \,+\,\mu_1\rho_*
                  \,-\, E(u_1,\P_*)
\intertext{and}
            p(\mu_2,u_2) 
            &\,>\, -\frac{1}{\beta}\,S_*
                  \,+\,\mu_2\rho_*
                  \,-\, E(u_2,\P_*)\,,
        \end{align*}
which gives
\[
   t\,p(\mu_1,u_1) \,+\, (1-t)\,p(\mu_2,u_2)
   \,>\, -\frac{1}{\beta}\,S_* \,+\, \mu\rho_* \,-\, E(u,P_*)
\]
because of the linearity of the specific interaction energy.
A comparison with \req{convexpressure-help} thus shows that the
pressure is strictly convex.

Consider now a fixed pair potential $u\in\U_0$.
It has been shown in the proof of \cite[Lemma~7.1]{HGEO94} 
that for any $\rho>0$ there exists a
translation invariant probability measure $\P_\rho$ on $\Gamma$ 
with density $\rho$, such that 
$S(\P_\rho)<\infty$, and $E(u,\P_\rho) < \infty$. 
Choosing $\P_*=\P_\rho$ in \req{Georgii}, Georgii's Gibbs variational principle 
(Theorem~\ref{Thm:Georgii}) yields the inequality
\be{Theta}
   p(\mu,u) - \mu\rho\,\geq\, -\frac{1}{\beta}S(\P_\rho)-E(u,\P_\rho) 
   \,=:\, c_\rho \,>\, -\infty
\ee
for every $\mu\in\R$. In other words,
\[
   \frac{p(\mu,u)}{\mu} \,\geq\, \rho \,+\, \frac{c_\rho}{\mu}\,,
\]
and hence,
\[
   \liminf_{\mu\to+\infty} \frac{p(\mu,u)}{\mu}\,\geq\, \rho\,.
\]
Since $\rho>0$ has been arbitrary, this implies \req{Thm:convexpressure}.
\end{proof}

We thus have shown that the relative entropy~\req{Georgii} is 
the sum of a strictly convex funtional and an affine function of $(\mu,u)$. 
Accordingly, the relative entropy is also a strictly convex functional
of $\mu$ and $u$, as long as it is finite. 
%

Another immediate consequence of Theorem~\ref{Thm:gvp}
is the following inequality for the pressure.

\begin{corollary}{\label{corr:subd}}
     For $\mu_1,\mu_2 \in \R$ and $u_1,u_2\in\U$ there holds  
            \be{pressurebounds}
            \begin{aligned}
                &(\mu_2-\mu_1) \rho_1-\frac{1}{2}\int_{\R^d} 
                (u_2-u_1)(x)\rho^{(2)}_{1}(x)\dx \\ 
                &\qquad \,\leq\, p(\mu_2,u_2)-p(\mu_1,u_1)
                 \,\leq\,  (\mu_2-\mu_1)\rho_{2}-\frac{1}{2}\int_{\R^d} 
                 (u_2-u_1)(x)\rho^{(2)}_{2}(x)\dx. 
            \end{aligned}
            \ee
            whenever $\rho_i$ and $\rho^{(2)}_i$ are the density and pair 
            correlation function of a $(\mu_i,u_i)$-Gibbs measure, respectively.
            Both inequalities are strict, unless $\mu_1=\mu_2$ and $u_1=u_2$.
            In particular, if $u_1=u_2$ and $\mu_1<\mu_2$, then
            $\rho_1<\rho_2$.
\end{corollary}

\begin{proof}
The two inequalities~\req{pressurebounds} follow readily from 
Theorem~\ref{Thm:gvp} by choosing for $\P_*$ the corresponding 
$(\mu_i,u_i)$-Gibbs measures, respectively. They are strict, 
unless $\mu_1=\mu_2$ and $u_1=u_2$.
\end{proof}

\section{The relative entropy functional for fixed density}
\label{Sec:rhofixed}
Returning to the inverse Henderson problem formulated in the introduction 
we now constrain our model ensembles to have the same density
$\rho_*$ as the target ensemble. 
Since the attainable densities for hard-core potentials are bounded we
need to distinguish the case whether 
$u$ is a hard-core potential or not. We focus our analysis
on the latter case and mention the necessary modifications for 
hard-core potentials in Remark~\ref{Rem:hardcore} later in this section.

The first fundamental problem to settle concerns the question whether and how
the prescribed density $\rho_*$ can be attained by some 
$(\mu,u)$-Gibbs measure for a given $u\in\U_0$.
When the chemical potential is sufficiently small,
i.e., when the system is in the \emph{gas phase}
(see Section~\ref{Sec:derivatives} for a specification of this term), 
then it is known that there
is a one-to-one relation between the corresponding chemical potentials 
and the associated densities, and in this case the inverse map 
$\rho\mapsto\mu$ can even be computed by means of 
cluster expansions, cf., e.g., Jansen, Kuna, and Tsagkarogiannis~\cite{JKT19}.
Outside the gas phase, when phase transitions may occur,
the problem becomes more difficult. 
Adopting a method from Chayes and Chayes~\cite{JTC83}
we can establish the following result, which holds for the full range of
possible densities.


\begin{theorem}
\label{Thm:Chayes}
Let $u\in \U_0$ and $\rho_*>0$ be fixed. Then there is a
unique chemical potential $\mu_*=\mu_*(u)\in\R$, for which there exists 
a $(\mu_*,u)$-Gibbs measure with density $\rho_*$.
\end{theorem}

\begin{proof}
For every $\mu\in\R$ and the given $u\in\U_0$
let $\P_{\mu,u}$ be a $(\mu,u)$-Gibbs measure,
and denote by $p(\mu)$ and $\rho(\mu)$ the associated pressure and density, 
respectively.

It is well-known, cf.~\cite[Theorem 4.3.1]{Ruel69}, that for small chemical
potentials
the pressure is a differentiable function with
\be{pprime}
   p'(\mu) \,=\, \rho(\mu)\,.
\ee
In other words, if $\rho_*$ is sufficiently small then the corresponding 
chemical potential $\mu_*$ from the formulation of the theorem is given 
as the unique minimum of the function
\be{theta}
   \Theta(\mu) \,=\, p(\mu)-\mu\rho_*\,,
\ee
the latter being strictly convex by virtue of Theorem~\ref{Thm:convexpressure}.
We will proceed by showing that the minimizer of $\Theta$ is also the 
appropriate chemical potential to choose for larger values of $\rho_*$.

To see this we first observe that
\[
   \Theta(\mu) \,\geq\, -\mu\rho_*\,,
\]
because the pressure is nonnegative, whereas
\be{rhostarplus}
   \Theta(\mu) \,\geq\, c_{(1+\eps)\rho_*} + \,\eps\mu\rho_*
\ee
for any suitable $\eps>0$ and corresponding constant
$c_{(1+\eps)\rho_*}$ by virtue of \req{Theta} (with $\rho=(1+\eps)\rho_*$).
This shows that $\Theta$ is bounded from below and that 
\[
   \Theta(\mu) \,\to\, +\infty\,, \qquad \text{whenever $|\mu|\to\infty$}\,.
\]
Therefore $\Theta$ attain its minimum for a uniquely defined value
$\mu=\mu_*$.

For any $\mu\in\R$, $\mu\neq\mu_*$, we now conclude from \req{pressurebounds} 
and \req{theta} that
\[
   (\mu_*-\mu)\rho(\mu) \,<\, p(\mu_*)-p(\mu)
   \,=\, \Theta(\mu_*) + \mu_*\rho_* - \Theta(\mu) - \mu\rho_*\,,
\]
i.e., 
\[
   (\mu_*-\mu)\bigl(\rho(\mu)-\rho_*\bigr) 
   \,<\, \Theta(\mu_*)-\Theta(\mu) \,<\,0\,.
\]
This means that
\be{enclosure}
   \begin{array}{c}
   \rho(\mu) \,<\, \rho_* \qquad \text{for } \ \mu<\mu_*\,, \\[1ex]
   \rho(\mu) \,>\, \rho_* \qquad \text{for } \ \mu>\mu_*\,.
   \end{array}
\ee

Now let $(\mu_k^+)_k$ be a strictly decreasing sequence and 
$(\mu_k^-)_k$ a strictly increasing sequence of chemical potentials,
both of which converge to $\mu_*$
By virtue of Lemma~\ref{Lem:gibbscontinuity} from the appendix there
exist $(\mu_*,u)$-Gibbs measures $\P^-$ and $\P^+$ with densities
\[
   \rho(\P^-) \,=\, \lim_{k\to\infty} \rho(\mu_k^-) \qquad \text{and} \qquad
   \rho(\P^+) \,=\, \lim_{k\to\infty} \rho(\mu_k^+)\,.
\]
The inequalities~\req{enclosure} imply that
\[
   \rho(\P^-) \,\leq\, \rho_* \,\leq\, \rho(\P^+)\,,
\]
and hence, there is some $t\in[0,1]$, for which
$\P=t\P^-+(1-t)\P^+$ has density
\[
   \rho(\P) \,=\, t\rho(\P^-) + (1-t)\rho(\P^+) \,=\, \rho_*\,.
\]
Since the set of $(\mu_*,u)$-Gibbs measures is convex,
$\P$ has all the desired properties from the statement of this
theorem, and the proof is done.
\end{proof}

In the light of the above theorem we can now restrict our attention to
$(\mu,u)$-Gibbs measures with density $\rho_*$, i.e., with
chemical potential $\mu=\mu_*(u)$, when looking at the relative entropy
functional. Further, we drop the constant offset $S_*$ in \req{Georgii},
as it is independent of $u$. This leads to the functional
\be{Phi}
    \Phi(u) = p_*(u) - \mu_*(u) \rho_* + E(u,\P_*)
\ee
for $u\in\U_0$, where we have set
\be{pstar}
   p_*(u) \,=\, p(\mu_*(u),u)
\ee
for brevity. By a slight abuse of wording we will keep calling $\Phi$
the relative entropy functional.
As we will see next, although $p_*$ may fail to be convex,
in general, the functional $\Phi$ is strictly convex, again.

\begin{theorem}
\label{Thm:stillstrictlyconvex}
        Let $\P_*$ be as in Theorem~\ref{Thm:Georgii}.
	    Then the functional $\Phi :\U_0 \to \R$ of \req{Phi} is  strictly convex
	    as long as the specific interaction enery $E(u,\P_*)$ is finite.
\end{theorem}
\begin{proof}
Let $u_1$ and $u_2$ be two different pair potentials from $\U_0$. 
For any fixed $0<t<1$ define
\[
   u \,=\, tu_1 \,+\, (1-t)u_2\,,
\]
and, as in Theorem~\ref{Thm:Chayes}, denote by $\mu=\mu_*(u)$ 
the chemical potential associated with $u$ and density $\rho_*$. 
Finally, let $\rho^{(2)}$  be
the pair correlation function of the associated $(\mu,u)$-Gibbs measure 
constructed in Theorem~\ref{Thm:Chayes}.

Then we obtain from \req{pstar} and \req{pressurebounds} 
-- with $\mu_i=\mu_*(u_i)$ for $i=1,2$ -- that
\begin{align*}
   &p_*(u) - t p_*(u_1)- (1-t)p_*(u_2) \\[1ex]
   &\quad =\, t \bigl(p_*(u)-p_*(u_1)\bigr)
               \,+\, (1-t)\bigl(p_*(u)-p_*(u_2)\bigr)
               \\[1ex]
   &\quad < \,t\left((\mu-\mu_1)\rho_*
                 \,-\, \frac{1-t}{2}\int_{\mathbb{R}^d} 
                           (u_2-u_1)(x)\rho^{(2)}(x) \dx\right) \\
   &\quad\phantom{=}\,
      + (1-t)\left((\mu -\mu_2)\rho_*
        \,-\, \frac{t}{2}\int_{\mathbb{R}^d} (u_1-u_2)(x)\rho^{(2)}(x)\dx
        \right)  \\[1ex]
   &\quad =\,\mu\rho_* \,-\ t\mu_1\rho_* \,-\, (1-t)\mu_2\rho_*\,.
	    \end{align*}
Note that this inequality is strict because $u$ is different from $u_1$ and
$u_2$ by construction. Reordering terms we thus arrive at
\be{phiconvex}
   p_*(u)- \mu\rho_*
   \,<\, t \bigl(p_*(u_1)-\mu_1\rho_*\bigr)
            \,+\,(1-t)\bigl(p_*(u_2)-\mu_2\rho_*\bigr)\,.
\ee
It thus follows from \req{Phi} and \req{phiconvex} that if the specific
interaction energy $E(\,\cdot\,,\P_*)$ stays finite then $\Phi$ is strictly
convex, because $E(\,\cdot\,,\P_*)$ is linear in the first argument.
\end{proof}

Theorem~\ref{Thm:stillstrictlyconvex} implies that the relative entropy
functional has at most one local minimizer, which is then also a global one. 
Concerning this minimizer we have the following result.

\begin{theorem}
\label{Thm:Phimin}
Let $\P_*$ be a $(\mu,u_*)$-Gibbs measure for some $u_*\in\U_0$ 
and $\mu\in\R$, and let $\rho_*$ be its density. 
Then the relative entropy function~\req{Phi} attains its minimum for $u=u_*$.
\end{theorem}

\begin{proof}
The given Gibbs measure $\P_*$ has finite specific entropy $S_*$, and hence,
\req{Georgii} implies that
\be{Thm:Phimin-help}
   \Phi(u) \,=\, \frac{1}{\beta}\,\Srel\bigl(\mu_*(u),u\bigr)
                 \,-\, \frac{1}{\beta}\,S_*
\ee
for every $u\in\U_0$. Since $\P_*$ has density $\rho_*$ the chemical
potential $\mu$ associated with $\P_*$ must be given by $\mu=\mu_*(u_*)$
according to Theorem~\ref{Thm:Chayes}. From \req{Thm:Phimin-help} and the 
Gibbs variational principle 
(Theorem~\ref{Thm:Georgii} and Theorem~\ref{Thm:gvp}) therefore follows that
\be{Thm:Phimin}
   \Phi(u_*) \,=\, -\,\frac{1}{\beta}\,S_*
   \,<\, \frac{1}{\beta}\,\Srel\bigl(\mu_*(u),u\bigr) 
         \,-\, \frac{1}{\beta}\,S_* \,=\, \Phi(u)
\ee
for every $u\in\U_0\setminus\{u_*\}$, and this was to be shown.
\end{proof}

Note from \req{Thm:Phimin} that the minimal value of $\Phi$ depends on
the specific entropy of the target Gibbs measure, and is therefore unknown
in general.

\begin{remark}
\label{Rem:hardcore}
\rm
For a hard-core potential $u\in\U$ the density $\rho$ of any 
associated $(\mu,u)$-Gibbs measure is bounded from above by
a finite \emph{closest-packing density} $\rcp=\rcp(u)$: This bound only 
depends on the hard-core radius $r_0$ of \req{potentials} and is a decreasing 
function of $r_0$, cf.~\cite[Section 7]{HGEO94}.
We formally set $\rcp(u)=+\infty$ for $u\in\U_0$.

If the given density $\rho_*$ of the target happens to be below $\rcp(u)$
for a given $u\in\U$, then Theorem~\ref{Thm:Chayes} is still valid for this
pair potential; 
to establish \req{rhostarplus} in its proof, $\eps>0$ must 
be so small that $(1+\eps)\rho_*$ is still below $\rcp(u)$, because this
is needed for \cite[Lemma~7.1]{HGEO94}, and hence, for \req{Theta}.
The relative entropy functional~\req{Phi}, however, is only 
well-defined on the domain
\[
   \dom\Phi \,=\, \{\,u\in\U:\, \rcp(u) > \rho_*\,\}\,,
\]
but this is again a convex set, compare Remark~\ref{Rem:convex},
and $\Phi$ happens to be strictly convex on $\dom\Phi$. 
Theorem~\ref{Thm:Phimin} extends literally to every $u_*\in\U$ with
this understanding of the domain of $\Phi$.
\fin
\end{remark}

Theorem~\ref{Thm:Phimin} is the basis for a variational setting of the
inverse Henderson problem: If the density and the pair correlation function
of a $(\mu_*,u_*)$-Gibbs measure target 
$\P_*$ are given, then the unique minimizer of the strictly convex 
relative entropy functional~\req{Phi} with $E(u,\P_*)$ of \req{EuP}
yields the corresponding pair potential $u_*$, 
i.e., the solution of the inverse Henderson problem. 
Nevertheless, this approach also has some pitfalls
as discussed in the following remark.

\begin{remark}
\label{Rem:Phimin}
\rm
If the target $\P_*$ fails to be some $(\mu,u)$-Gibbs measure
then it is not clear to us whether the relative entropy functional $\Phi$ 
will still be bounded from below on $\U$,
and even if it may, its infimum need not be attained on $\U$.

Vice versa, if $u_*\in\U$ happens to be the miminizer of $\Phi$, 
then this does not imply that $\P_*$ is a $(\mu_*(u_*),u_*)$-Gibbs measure. 
To see the latter, consider the following example: 
Let $u_*$ be any hard-core potential in $\U$ and $\rho_*<\rcp(u_*)$;
compare Remark~\ref{Rem:hardcore}.
Then, provided $\rho_*$ is sufficiently small,
a result by Kuna, Lebowitz, and Speer~\cite[Corollary~4.3]{Kuna7} states
that there exist uncountably many distinct 
translation invariant probability measures with finite specific entropy
and density $\rho_*$, which share the pair correlation function $\pcf_*$
with the $(\mu_*(u_*),u_*)$-Gibbs measure $\P$, but $\P$ is the only 
$(\mu,u)$-Gibbs measure among them according to the Henderson theorem.
Because the specific interaction energy is given by \req{EuP} for all of these
probability measures, the relative entropy functional does not
differ, and hence, $u_*$ is its unique minimizer, regardless which of them
has been the target $\P_*$.

Finally we mention that even if the target $\P_*$ is a 
$(\mu,u_*)$-Gibbs measure,
and hence, the functional~\req{Phi} is minimized by $u_*$ according to
Theorem~\ref{Thm:Phimin}, then this does not seem to imply that the model 
and the target have the same pair correlation function. 
To illustrate this, imagine that a fluid corresponding to a pair potential
$u_*\in\U$ exhibits a so-called \emph{triple point} 
(compare, e.g., \cite{HaMcD13}), 
where three different phases coexist at the same thermodynamical state point,
i.e., for the same values of pressure (or chemical potential $\mu_*$, say) 
and temperature (or inverse temperature $\beta$). It can be expected
that the different phases have linearly independent pair correlation
functions; taking convex combinations of the corresponding Gibbs measures
one can thus determine two $(\mu_*,u_*)$-Gibbs measures which 
have the same density $\rho_*$, but have different pair correlation functions.
One of them could be the target and the other one the minimizing model. 
We will see in Section~\ref{Sec:derivatives} that if the density $\rho_*$
belongs to the gas phase of the minimizer of $\Phi$ then the pair
correlation function of the minimizing Gibbs measure always coincides
with the given $\pcf_*$;
see Remark~\ref{Rem:Phimin2}.
\fin
\end{remark}

\section{Differentiability of the relative entropy functional}
\label{Sec:derivatives}
In the remainder of this paper we further analyze the case when $\Phi$
is a differentiable function of $u$. For this we build upon our earlier
work~\cite{Hank18a,Hank18b}.\footnote{The results in \cite{Hank18a,Hank18b} 
are formulated in three-dimensional space (i.e., for $d=3$) and for 
pair potentials, which are rotation invariant;
also, hard-core potentials had not been considered.
All results extend verbatim to the setting considered here.}
A conceptual difficulty in this context is 
the fact that $\U$ lacks a universal topology. Therefore,
following \cite{Hank18a}, we consider a tailored neighborhood for any 
given $u_0\in\U$ by 
introducing a corresponding Banach space $\V$ of \emph{perturbations}, 
consisting of all even measurable functions $v\colon\R^d\to\R$,
for which the associated norm
\be{Vd}
   \norm{v}_\V \,=\, \sup_{x\in\R^d} \frac{|v(x)|}{\psi_0(|x|)} 
\ee
is finite. Here, $\psi_0$ is the majorant $\psi$ of \req{varphi,psi} 
associated with $u_0$. Note that the norm~\req{Vd} is somewhat stronger than 
the one employed in \cite{Hank18a} because it is more restrictive in the
core region $0<r< r_0$, and hence,
the resulting space of perturbations is smaller.
We believe that this restriction allows a more natural formulation 
of our results.

\begin{remark}
\rm
It can easily be seen that for any  $v\in\V$ the sum $u_0+v$ belongs to $\U$:
If $u_0$ is a hard-core potential, then $r_0$ is its hard-core radius and
the corresponding function $\varphi$ from \req{varphi,psi} equals $+\infty$,
and we have
\[
\begin{array}{ll}
    \phantom{\bigl|\bigr|}
    u_0(x)+v(x) \,\geq\, \varphi(|x|)
    &\quad\text{ for } 0 < |x| < r_0\,, \\[1ex]
    \bigl|u_0(x)+v(x)\bigr| \,\leq\, \bigl(1\,+\,\norm{v}_\V\bigr)\psi_0(|x|)
    &\quad\text{ for } |x| \geq r_0\,.
\end{array}
\]
Otherwise, there exists $r_1\in(0,r_0]$ such that 
$\varphi(r) \geq \norm{v}_\V\psi_0(0)$ for $0 < r \leq r_1$, and therefore
\[
\begin{array}{ll}
    \phantom{\bigl|\bigr|}
    u_0(x)+v(x) \,\geq\, \varphi(|x|) - \norm{v}_\V\psi_0(0) 
    &\quad\text{ for } 0 < |x| < r_1\,, \\[1ex]
    \bigl|u_0(x)+v(x)\bigr| \,\leq\, 
    {\displaystyle
    \Bigl(\sup_{|x'|\geq r_1}\frac{|u(x')|}{\psi_0(|x'|)}
          \,+\,\norm{v}_\V\Bigr)\psi_0(|x|)}
    &\quad\text{ for } |x| \geq r_1\,.
\end{array}
\]
In either case this shows that $u_0+v\in\U$.
\fin
\end{remark}

We now consider the ball
\be{BV}
   \BV(u_0) \,=\, \bigl\{ u=u_0+v\,:\, \norm{v}_\V<\delta_0\bigr\}
   \,\subset\,\U
\ee
around $u_0$ for a suitable radius $\delta_0$ with $0<\delta_0<1$;
we further specify $\delta_0$ in the context of \req{rho0} below. 
As has been verified in \cite[Proposition~2.1]{Hank18a}, 
there are positive constants $c_\beta$ and $B$,
such that every $u\in\BV(u_0)$ satisfies
\be{cbeta}
   \int_{\R^d} |e^{-\beta u(x)}-1|\dx \,\leq\, \cbeta\,,
\ee
and the interaction energy of every configuration $\gamma\subset\R^d$ with 
$N(\gamma)$ particles is bounded from below by
\be{B}
    U(\gamma) \,\geq\, -B\thsp N(\gamma)\,.
\ee
Moreover, for every $u\in\BV(u_0)$ and $\mu\in\R$ and every associated
$(\mu,u)$-Gibbs measure the sequence
\[
   \bfrho(\mu,u)\,=\,[\rho,\rho^{(2)},\dots]^T
\]
of its 
$m$-particle correlation functions is a solution of the so-called
\emph{Kirkwood-Salsburg equations} 
(cf., e.g., \cite[Corollary~5.3]{Ruel70}, or compare~\req{KS2} in the appendix),
which can be written in the form
\be{KS}
   \bigl(I-e^{\beta\mu} A(u)\bigr)\bfrho(\mu,u) \,=\, e^{\beta\mu} \bfe\,,
\ee
where $\bfe=[1,0,\dots]^T$, $I$ is the identity operator,
and $A(u)$ is a bounded operator in an associated 
Banach space $\X$ of sequences of $L^\infty$ functions. 
This operator $A(u)$ is given by a 
certain infinite-dimensional matrix of integral operators,
compare~\cite{Ruel69,Hank18a}, and is Fr\'echet differentiable with respect to
$u\in\BV(u_0)$.

We fix
\be{mu0}
   \mu_0 \,<\, \frac{1}{\beta}\,\log\frac{1}{c_\beta e^{2\beta B+1}}\,,
\ee
and call the range $\mu\leq\mu_0$ of chemical potentials the \emph{gas phase} 
associated with the pair potentials in $\BV(u_0)$. 
In this gas phase the Kirkwood-Salsburg equations~\req{KS} have a unique 
solution $\bfrho(\mu,u)$, which can be
developed in a converging Neumann series in $\X$; in particular, this means
that for $u\in\BV(u_0)$ there exists only one $(\mu,u)$-Gibbs measure for each 
chemical potential $\mu$ within the gas phase.
The individual components $\rho^{(m)}(\mu,u)$ of $\bfrho(\mu,u)$ are
Fr\'echet differentiable (in $L^\infty$) with respect to $u$, 
analytic with respect to $\mu<\mu_0$ and continuous in $\mu\leq\mu_0$.
It is also easy to see that $\bfrho$ is a $C^1$ function of $\mu\leq\mu_0$
and $u\in\BV(u_0)$.

Further, for fixed $u\in\BV(u_0)$ the density $\rho=\rho(\mu,u)$ 
of the associated Gibbs measure, i.e., the first entry of $\bfrho(\mu,u)$, 
is differentiable and strictly increasing in $\mu$ up to $\mu=\mu_0$ 
by virtue of Corollary~\ref{corr:subd}. 
Since it is also continuous in $u\in\BV(u_0)$ we can choose
$\delta_0$ in \req{BV} so small that
\be{rho0}
   \rho_0
   \,=\, \inf \bigl\{ \rho(\mu_0,u)\,:\, u\in\BV(u_0)\bigr\} 
   \,\geq\, (1-\eps)\,\rho(\mu_0,u_0)
\ee
for any given $\eps>0$, and then every density $\rho\in(0,\rho_0)$ belongs
to the gas phase of all pair potentials in $\BV(u_0)$.
Finally, for fixed $u\in\BV(u_0)$, the pressure $p(\mu,u)$ is also 
differentiable and strictly increasing in $\mu$ up to $\mu=\mu_0$, and  
its derivative is given by the density, cf.~\req{pprime}.


\begin{proposition}
\label{Prop:mustardiff}
Let $0<\rho_*<\rho_0$ with $\rho_0$ as in \req{rho0}. 
Then the chemical potential $\mu_*=\mu_*(u)$
defined in Theorem~\ref{Thm:Chayes} is differentiable with respect to 
$u\in\BV(u_0)$ with derivative $\mu_*'(u)\in\V'$ given by
\be{mustardiff}
    \mu_*'(u)\,v 
    \,=\, - \frac{\partial_u\rho(\mu_*(u),u)\,v}
                 {\partial_\mu\rho(\mu_*(u),u)}\qquad
    \text{for $v\in\V$},
\ee
where $\partial_\mu\rho$ and $\partial_u\rho$ denote the 
partial derivatives of $\rho(\mu,u)$, and $\V'$ is the dual space of $\V$.
\end{proposition}

\begin{proof}
Since the pressure and the density of a fixed pair potential $u\in\BV(u_0)$
are strictly increasing and differentiable functions of the chemical potential
in the gas phase, we can rewrite each of these functions as a strictly 
increasing function of any of the other ones. The pressure, for example,
can be written as a function of density, which we call $\pi$ to avoid any
confusion, i.e.,
\[
   p(\mu,u) \,=\, \pi\bigl(\rho(\mu,u),u\bigr)\,.
\]
Then the chain rule can be applied to obtain
\be{chainrule}
   \partial_\mu p(\mu,u) \,=\, 
   \partial_\rho\pi\bigl(\rho(\mu,u),u\bigr)
   \partial_\mu \rho(\mu,u)\,,
\ee
because it has been shown in \cite[Theorem~4.3]{Ruel70} 
that $\pi$ is a differentiable function of the density in the gas phase.
For $\mu=\mu_*(u)$ we thus conclude from \req{pprime} and \req{chainrule} that
\be{alphaF}
   \rho_* \,=\, \rho\bigl(\mu_*(u),u\bigr) 
   \,=\, \partial_\mu p\bigl(\mu_*(u),u\bigr)
   \,=\, \alpha\,\partial_\mu \rho\bigl(\mu_*(u),u\bigr)
\ee
with 
\[
   \alpha \,=\, \partial_\rho \pi(\rho_*,u)\,>\, 0\,.
\]

It thus follows from \req{alphaF} that
$\partial_\mu\rho\bigl(\mu_*(u),u\bigr)>0$, so that  
the implicit function theorem is applicable to the equation
\[
   \rho(\mu_*(u+v),u+v) \,=\, \rho_* 
\]
near $v=0$.
From this we readily obtain \req{mustardiff};
moreover, $\mu_*'(u)$ belongs to $\V'$ because $\partial_u\rho\in\V'$.
\end{proof}

Now we return to our analysis of the relative entropy functional~\req{Phi}.

\begin{theorem}
\label{Thm:Phiprime}
Let $u_0\in\U$ be a fixed pair potential and $\BV(u_0)\subset\U$ 
be defined as in \req{BV} with $\delta_0$ so small that \req{rho0} holds.
Furthermore, let the target $\P_*$ of the relative entropy functional
have density $\rho_*<\rho_0$, and assume that the associated
specific interaction energy $E(\,\cdot\,,\P_*)$ 
is bounded on $\BV(u_0)$.
Then $\Phi$ is
Fr\'echet-differentiable in every $u\in\BV(u_0)$ with derivative 
$\Phi'(u)\in\V'$, given by
\be{Phiprime}
    \Phi'(u)\thsp v \,=\, E(v,\P_*) \,-\, \frac{1}{2}\int_{\R^d}v(x)\pcf(x)\dx
\ee
for $v\in\V$,
where $\rho^{(2)}$ is
the pair correlation function of the associated $(\mu_*(u),u)$-Gibbs measure.
\end{theorem}

\begin{proof}
Let $u\in\BV(u_0)$ and $v\in\V$.
We apply Corollary~\ref{corr:subd} to $u_1=u$ and $u_2=u+v$ with
$\mu_1=\mu_*(u)$ and $\mu_2=\mu_*(u+v)$. 
This yields
\begin{align*}
   - \frac{1}{2}\int_{\R^d} v(x)\rho^{(2)}(x)\dx
   &\,\leq\, 
   p_*(u+v)-\mu_*(u+v)\rho_* \,-\, \bigl(p_*(u) -\mu_*(u)\rho_*\bigr)\\[1ex]
   &\,\leq\, -\frac{1}{2}\int_{\R^d} 
            v(x)\trho^{(2)}(x)\dx\,,
\end{align*}
where $\trho^{(2)}$ denotes the pair correlation function of the
corresponding $(\mu_*(u+v),u+v)$-Gibbs measure. 
It follows that
\begin{align*}
   E(v,\P_*) \,-\, \frac{1}{2}\int_{\R^d} v(x)\pcf(x)\dx
   &\,\leq\, 
   \Phi(u+v) \,-\, \Phi(u)\\[1ex]
   &\,\leq\, E(v,\P_*) \,-\, \frac{1}{2}\int_{\R^d} 
            v(x)\trho^{(2)}(x)\dx\,,
\end{align*}
i.e., that
\begin{align*}
    0 &\,\leq\, \Phi(u+v)-\Phi (u)
      \,-\,\Bigl(E(v,\P_*) - \frac{1}{2}\int_{\R^d}v(x)\pcf(x)\dx\Bigr) \\[1ex]
    &\,\leq\, \frac{1}{2} \int_{\R^d}
    v(x)\big(\rho^{(2)}-\trho^{(2)}\big)(x)\dx\,.
\end{align*}
This shows that
\begin{align*}
    &\left|\Phi(u+v)-\Phi (u)
      \,-\,\Bigl(E(v,\P_*) - \frac{1}{2}\int_{\R^d}v(x)\pcf(x)\dx\Bigr) 
      \right|\\[1ex]
    &\qquad\,\leq\,  \frac{C_{\psi_0}}{2}\,\norm{v}_{\V} 
            \norm{\trho^{(2)}-\rho^{(2)}}_{L^\infty(\R^d)}\,,
\end{align*}
where
\[
   C_{\psi_0} \,=\, |\S^{d-1}| \int_0^\infty r^{d-1}\psi_0(r)\dr
   \,<\, \infty
\]
by virtue of \req{varphi,psi}. We finally observe that
\[
   \trho^{(2)} \,=\, \rho^{(2)}(\mu_*(u+v),u+v)
\]
converges to $\rho^{(2)}=\rho^{(2)}(\mu_*(u),u)$ in $L^\infty(\R^d)$
as $\norm{v}_\V\to 0$
because of Proposition~\ref{Prop:mustardiff} and the smoothness of $\rho^{(2)}$
as a function of $\mu$ and $u$. We have therefore established that
\begin{align*}
   \left|\Phi(u+v)-\Phi (u)
      \,-\,\Bigl(E(v,\P_*) - \frac{1}{2}\int_{\R^d}v(x)\pcf(x)\dx\Bigr) 
      \right|
   \,=\, o(\norm{v}_\V)
\end{align*}
as $\norm{v}_\V\to 0$, which yields \req{Phiprime}. 
Since $E(\,\cdot\,,\P_*)$ is bounded on $\BV(u_0)$, it is easy to see that
$E(\,\cdot\,,\P_*)\in\V'$, and hence, it follows  
in the same way as above that $\Phi'(u)\in\V'$ with
\[
   \norm{\Phi'(u)}_{\V'} \,\leq\, \norm{E(\,\cdot\,,\P_*)}_{\V'}
   \,+\, \frac{C_{\psi_0}}{2}\,\norm{\pcf}_{L^\infty(\R^d)}\,.
\]
This concludes the proof.
\end{proof}

\begin{remark}
\label{Rem:nablaPhi}
\rm
We emphasize that the dual space $\V'$ of $\V$ is a complicated Banach space. 
However, the representation~\req{Phiprime}, in combination with \req{EuP}, 
reveals that if $\P_*$ has a bounded and measurable pair correlation function 
$\rho^{(2)}_*$,
then $\Phi'(u)$ can be identified with 
\be{nablaPhi}
   \nabla\Phi(u) \,=\, 
   \frac{1}{2}\bigl(\rho_*^{(2)}-\rho^{(2)}(\mu_*(u),u)\bigr)
   \,\in\, L^\infty(\R^d)\,,
\ee
when using the natural dual pairing
\be{scalp}
   \Phi'(u)v \,=\, \scalp{v,\nabla\Phi(u)}
   \,=\, \int_{\R^d} v(x)\bigl(\nabla\Phi(u)\bigr)(x)\dx
\ee
of $L^2(\R^d)$.

This is the case,
for example, if $\P_*$ is a $(\mu,u_*)$-Gibbs measure for some 
$u_*\in\U$ (with $\mu=\mu_*(u_*)$), and then 
\req{nablaPhi} can further be rewritten as
\[
   \nabla\Phi(u) \,=\, 
   \frac{1}{2}\bigl(\omega^{(2)}(\mu,u_*) -\omega^{(2)}(\mu_*(u),u)\bigr)\,,
\]
where 
\be{omega2}
   \omega^{(2)}(\mu,u) \,=\, \pcf(\mu,u) - \rho(\mu,u)^2\,.
\ee
It follows from \req{Thm448} below that $\nabla\Phi(u)$ then also belongs
to $L^1(\R^d)$, and hence, to $L^2(\R^d)$ as well.
Note that $\V$ is embedded in $L^2(\R)$ for the same reason.
We refer to Section~\ref{Sec:L2} for further arguments which 
support the choice of the $L^2$-pairing.
\fin
\end{remark}

\begin{remark}
\label{Rem:Phimin2}
\rm
Assume that the target $\P_*$ has a bounded and measurable 
pair correlation function $\pcf_*$, 
and that the relative entropy functional $\Phi$ is minimized by some 
$u_*\in\U$. Assume further that the given density $\rho_*$ belongs to 
the gas phase of $u_*$. Then it follows from Theorem~\ref{Thm:Phiprime}
that $\Phi$ is differentiable at $u_*$, and that $\nabla\Phi(u_*) = 0$. 
From \req{nablaPhi} we thus conclude 
that the target and the model share the same pair correlation
function in this case; compare Remark~\ref{Rem:Phimin}.
\fin
\end{remark}

\begin{remark}
\label{Rem:Ursell}
\rm
For a $(\mu,u)$-Gibbs measure
the \emph{cluster functions}, sometimes also called Ursell functions
(e.g., Stell~\cite{Stel64}, and \cite{Hank18b}), or truncated correlation 
functions (e.g. Jansen~\cite{Jansen19}),
are defined recursively by the constant function $\omega^{(1)} =\rho$ and,
for $m\geq 2$, by
\be{omegan}
   \omega^{(m)}(x_1,\dots,x_m)
   \,=\, \rho^{(m)}(x_1,\dots,x_m) \,-\, \sum_{k=2}^m\, \sum_{\pi \in\Pi_k^m}\,
   \prod_{i=1}^k \omega^{(|\pi_i|)}(\xx_{\pi_i})\,,
\ee
cf.~\cite[p.~87]{Ruel69},
where $\Pi_k^m$ denotes the set of partitions $\pi$ of $x_1,\dots,x_m$ into 
$k$ subsets $\xx_{\pi_1},\dots,\xx_{\pi_k}$. 
Note that it immediately follows from \req{omegan} that each cluster function 
inherits the translation and permutation invariance of the correlation 
functions. 

For $m=2$ we recover from \req{omegan} the definition~\req{omega2}, 
and for fixed $\mu$ and $u$ we adopt the short-hand notation
\[
   \omega^{(2)}(x) \,=\, \rho^{(2)}(x) \,-\, \rho^2
\]
for $\omega^{(2)}(x,0)$ from the pair correlation function.
Note that 
\be{om2even}
   \omega^{(2)}(x) \,=\, \omega^{(2)}(-x) \qquad \text{for every $x\in\R^d$}
\ee
because of the translation invariance.

We will make repeatedly use of \cite[Theorem~4.4.8]{Ruel69} which states that
\be{Thm448}
   \int_{(\R^d)^{(m-2)}} \bigl|\omega^{(m)}(\,\cdot\,,0,x_3,\dots,x_m)\bigr|
   \dx_3\cdots\dx_m \,\in\, L^1(\R^d)\,,
\ee
provided that $u\in\U$ and $\mu\leq\mu_0$; see also Lemma~\ref{Lem:Frommer}
in the appendix. 
In particular, \req{Thm448} shows that $\omega^{(2)}\in L^1(\R^d)$.
\fin
\end{remark}

Based on Theorem~\ref{Thm:Phiprime} we now
compute the Hessian of $\Phi$.

\begin{theorem}
\label{Thm:Hessian}
Under the assumptions of Theorem~\ref{Thm:Phiprime} the relative
entropy functional~\req{Phi} is twice differentiable, and its Hessian 
at $u\in\BV(u_0)$ is given by
\be{Phipp}
   \Phi''(u)(v,\vt)
   \,=\, - \frac{1}{\partial_\mu\rho}
           \bigl(\partial_u \rho\,\vt\bigr)\bigl(\partial_u \rho\,v\bigr) 
         \,-\, \frac{1}{2}\thsp\scalp{v,\partial_u\pcf\,\vt}
\ee
in terms of the dual pairing~\req{scalp} of $L^2(\R^d)$,
where $v,\vt\in\V$, and the partial derivatives $\partial_u\rho$ and $\partial_u\pcf$
are evaluated at the pair potential $u$ and the chemical potential $\mu_*(u)$.
\end{theorem}

\begin{proof}
We differentiate \req{Phiprime} with respect to $u$ in the direction 
$\vt\in\V$. Keeping in mind that $\rho^{(2)}$ in \req{Phiprime} stands for 
$\rho^{(2)}(\mu_*(u),u)$ the chain rule gives
\be{Hessian}
\begin{aligned}
   \Phi''(u)(v,\vt)
   &\,=\, -\frac{1}{2}\int_{\R^d} v(x)
          \Bigl(\partial_\mu \pcf\mu_*'(u)\,\vt
                \,+\, \partial_u \pcf\,\vt\Bigr)(x)\dx\\[1ex]
   &\,=\, -\frac{\mu_*'(u)\,\vt}{2}\, \scalp{v,\partial_\mu\rho^{(2)}}
          \,-\, \frac{1}{2}\thsp \scalp{v,\partial_u\pcf\,\vt}\,,
\end{aligned}
\ee
where we have adopted the notation~\req{scalp},
and the partial derivatives of $\rho^{(2)}$ 
-- which belong to $L^\infty(\R^d)$ --
are evaluated at the pair potential $u$ and the chemical potential $\mu_*(u)$.

To determine $\partial_\mu\pcf(\mu,u)$ for an arbitrary
chemical potential $\mu<\mu_0$ we apply Corollary~\ref{corr:subd}
with $u_1=u$ and $u_2=u+w$ for any $w\in\V$, and with 
$\mu_1=\mu_2=\mu$: As in the proof of
Theorem~\ref{Thm:Phiprime} we conclude that the pressure is differentiable 
with respect to $u$ with derivative $\partial_u p(\mu,u)\in\V'$ given by
\[
    \partial_u p(\mu,u)\,w \,=\, -\frac{1}{2}\int_{\R^d}w(x)\pcf(x)\dx
    \,=\, -\frac{1}{2} \scalp{w,\pcf}\,, 
\]
where $\pcf$ is the pair correlation function of the $(\mu,u)$-Gibbs measure.
Schwarz's theorem and \req{pprime} therefore yield
\be{durho}
   \scalp{w,\partial_\mu \pcf} \,=\, \partial_\mu \scalp{w,\pcf}
   \,=\, -2\,\partial_\mu \bigl((\partial_u \,p)w\bigr)
   \,=\, -2\,\partial_u (\partial_\mu \,p)w
   \,=\, -2\,\partial_u \rho\,w\,.
\ee
Inserting this into \req{Hessian} we obtain
\begin{align*}
   \Phi''(u)(v,\vt)
   \,=\, \bigl(\mu_*'(u)\,\vt\bigr)\bigl(\partial_u\rho\, v)
          \,-\, \frac{1}{2}\scalp{v,\partial_u\pcf\,\vt}\,,
\end{align*}
and the assertion~\req{Phipp} thus follows from 
Proposition~\ref{Prop:mustardiff}.
\end{proof}

Note that $\Phi''(u)$ is a continuous bilinear form: This follows readily
from the result in \cite{Hank18a} that $\partial_u\rho\in\V'$ and 
$\partial_u\pcf\in\L\bigl(\V,L^\infty(\R^d)\bigr)$. We further observe
that $\Phi''$ is continuous in $u$ 
because $\rho$ and $\pcf$ are $C^1$ functions of $\mu$ and $u$.
This implies that $\Phi''(u)$ is a symmetric bilinear form. 
Finally, this bilinear form is positive semidefinite, because $\Phi$
is strictly convex.

\section{The connection to the IMC iterative scheme}
\label{Sec:IMC}
In the preceding section we have seen that the relative entropy functional
is twice differentiable at any $u_0\in\U$, provided that the target density
$\rho_*$ belongs to its gas phase. The convex quadratic approximation
\be{Phi2}
   \Phi_2(u_0+v) \,:=\, 
   \Phi(u_0) \,+\, \Phi'(u_0)v \,+\, \frac{1}{2}\,\Phi''(u_0)(v,v)
   \,\approx\, \Phi(u_0+v)\,,
\ee
valid for small enough $v\in\V$,
attains its minimum for every solution $v_0\in\V$ of the variational problem
\be{Newtonupdate}
   \Phi''(u_0)(v_0,w) \,=\, - \Phi'(u_0)w \qquad 
   \text{for all $w\in\V$},
\ee
and the well-known Newton scheme for minimizing $\Phi$ consists in
updating $u_0$ via
\be{Newtonstep}
   u_1 \,=\, u_0 \,+\, v_0
\ee
to obtain a (hopefully) better approximation of the global unique
minimizer of $\Phi$.
Unfortunately, however, we only know that $\Phi''(u_0)$ is semidefinite;
it is therefore not clear whether \req{Newtonupdate} has a solution, and 
if it does, whether it is unique.

The bilinear form $\Phi''(u_0)$ defines a linear operator $A:\V\to\V'$
via the identity 
\[
   \Phi''(u_0)(v,\vt) \,=\, \scalp{\vt,Av}_{\V\times\V'} \qquad 
   \text{for all $v,\vt\in\V$}\,.
\]
As in Remark~\ref{Rem:nablaPhi} the dual form on the right-hand side 
can be replaced by the dual pairing induced by $L^2(\R^d)$, 
when $A$ is identified with the (bounded) operator
\be{A}
   A:\V\to L^\infty(\R^d)\,, \qquad 
   A:v\,\mapsto\,
   -\frac{1}{2}\,\partial_u\rho^{(2)}v \,-\, 
   \frac{1}{\partial_\mu\rho}(\partial_u\rho\,v)\,\nabla_u\rho\,,
\ee
where $\nabla_u\rho\in L^\infty(\R^d)$ denotes
the representative of the functional $\partial_u\rho\in\V'$ for this pairing.
This is an immediate consequence of Theorem~\ref{Thm:Hessian}.
According to \cite[Section~6]{Hank18b} there holds
\be{nablarho-tmp}
   \nabla_u\rho(x) \,=\, -\beta \pcf(x) 
        \,-\, \frac{\beta}{2} \int_{\R^d} \chi^{(3)}(0,x+x',x')\dx'\,,
\ee
where
\[
   \chi^{(3)}(0,x+x',x') \,=\, \omega^{(3)}(0,x+x',x')
   \,+\, \rho\,\omega^{(2)}(x') \,+\, \rho\,\omega^{(2)}(x+x')
\]
is given in terms of the second and third cluster functions 
associated with $\mu$ and $u$; compare Remark~\ref{Rem:Ursell}.
Using the translation and permutation invariance of $\omega^{(3)}$ this can
be rewritten as
\be{chi3}
   \chi^{(3)}(0,x+x',x') \,=\, \omega^{(3)}(x,0,-x')
   \,+\, \rho\,\omega^{(2)}(x') \,+\, \rho\,\omega^{(2)}(x+x')\,.
\ee
Note that it has been shown in \cite[Proposition~4.2]{Hank18b} that 
the integral in \req{nablarho-tmp} is absolutely convergent and that it defines 
a bounded function of $x\in\R^d$. Since $\omega^{(2)}\in L^1(\R^d)$ it thus
follows from \req{chi3} that 
\be{om3int}
   \int_{\R^d} \bigl|\omega^{(3)}(\,\cdot\,,0,-x')\bigr| \dx'
   \,\in\, L^\infty(\R^d)\,,
\ee
and inserting~\req{chi3} into \req{nablarho-tmp}, we finally arrive at
the representation
\be{nablarho}
   \nabla_u\rho(x) \,=\, -\beta \pcf(x) 
        \,-\, \frac{\beta}{2} \int_{\R^d} \omega^{(3)}(x,0,x')\dx'
        \,-\, \beta\rho \int_{\R^d} \omega^{(2)}(x')\dx'\,,
\ee
which will be used later on.

Now we introduce the (nonlinear) ''Henderson operator''
\be{Henderson-Op}
   F:u \,\mapsto\, \pcf\bigl(\mu_*(u),u\bigr)\,,
\ee
which maps $u\in\BV(u_0)$ to its associated pair correlation function
at the given density $\rho_*$, and which is an injective operator
according to the Henderson theorem.
Since $F(u)+2\nabla\Phi(u)$ is independent of $u$ according to \req{Phiprime},
see also \req{nablaPhi}, it follows that $F:\BV(u_0)\to L^\infty(\R^d)$ is 
differentiable for $\rho_*<\rho_0$ with
\be{FprimeHessian}
   \scalp{\vt,F'(u_0)v} \,=\, -2\,\Phi''(u_0)(v,\vt)\,.
\ee
This shows that if the target $\P_*$ has a pair correlation function
$\pcf_*\in L^\infty(\R^d)$ then the variational problem~\req{Newtonupdate} 
is equivalent to the linear operator equation
\be{IMC}
   F'(u_0)v_0 \,=\, \pcf_* \,-\, F(u_0)\,.
\ee
The step~\req{Newtonstep} for minimizing $\Phi$ with Newton's method is
thus equivalent to one iteration of the IMC method, where the update $v_0$ is
determined by solving~\req{IMC}.
This is the thermodynamic limit analog of the observation
in \cite{MKV09,RSSV19}, referred to in the introduction.

\begin{remark}
\label{Rem:IMC}
\rm
As follows from Theorem~\ref{Thm:Phiprime} and \req{EuP} the Henderson operator
satisfies
\[
   \scalp{v,F(u_0)} \,=\, 2\,E(v,\P_0)\,,
\]
where $\P_0$ denotes the corresponding $(\mu_*(u_0),u_0)$-Gibbs measure.
In the IMC scheme, as introduced in \cite{LyLa95}, this
interaction energy is approximated by the corresponding expectation value
\[
   E(v,\P_0) \,\approx\, \frac{1}{|\Lambda|}\,\scalp{V}_{\rho_*,\Lambda}
\]
of the \emph{canonical ensemble} corresponding to the pair potential $u_0$ at 
density $\rho_*$ in the bounded box $\Lambda\subset\R^d$, where $V$ is the 
observable
\[
   V(\gamma) \,=\, \sum_{1\leq i<j\leq N} v(x_i-x_j)
\]
associated with the perturbation $v\in\V$, compare~\req{U}.
Likewise, the derivative $F'(u_0)$ is approximated in the IMC scheme via the
cross correlation matrix
\[
   -\frac{2\beta}{|\Lambda|}\Bigl(\scalp{V\Vt}_{\rho_*,\Lambda}
   - \scalp{V}_{\rho_*,\Lambda}\scalp{\Vt}_{\rho_*,\Lambda}\Bigr)\
   \,\approx\,\scalp{\vt,F'(u_0)v}\,.
\]

It has to be emphasized that the density constraint $\rho(u)=\rho_*$
is inherently built into the canonical ensemble. Accordingly, this
approximation of $F'(u_0)$ does respect this constraint. 
When working in the grand canonical ensemble instead
(or in the thermodynamic limit) this constraint necessitates a correction term
to project the unconstrained derivative 
into the tangent plane of the manifold of pair correlation functions 
corresponding to particle systems with the correct density.
This is the reason for the first term on the right-hand side of \req{Phipp} 
which does not (and need not) occur in the 
original IMC scheme.
\fin
\end{remark}


\section{Extension of the Hessian to a semidefinite bilinear form on  
\boldmath{$L^2(\R^d)$}}
\label{Sec:L2}
In the remainder of this paper 
we study the Hessian of the relative entropy functional somewhat further. 
More precisely, we examine the mapping properties of the 
Jacobian of the Henderson map, which is connected to the Hessian $\Phi''$
via \req{FprimeHessian}.

To begin with, recall from Remark~\ref{Rem:nablaPhi} that the right-hand 
side of \req{IMC} belongs to $L^\infty(\R^d)\cap L^1(\R^d)$,
when the target $\P_*$ is a $(\mu,u_*)$-Gibbs measure. As our first
result of this section we show that $F'(u_0)$ has matching mapping properties.

\begin{theorem}
\label{Thm:Fprime1}
Under the assumptions of Theorem~\ref{Thm:Phiprime}
the Jacobian $F'(u_0)$ of the Henderson operator~\req{Henderson-Op}
belongs to $\L(\V,L^\infty(\R^d))\cap\L(\V,L^1(\R^d))$.
\end{theorem}

\begin{proof}
We already know from Section~\ref{Sec:IMC} that 
$F'(u_0)=-2A\in\L(\V,L^\infty(\R^d))$.
It thus remains to establish that $F'(u_0)\in\L(\V,L^1(\R^d))$. 
To this end we rewrite the Henderson operator in terms of the 
cluster function~\req{omega2}, i.e.,
\[
   F(u_0) \,=\, \omega^{(2)}(\mu_*(u_0),u_0)\,+\, \rho_*^2\,,
\]
which yields
\be{Fprime-omega}
   F'(u_0)\thsp v 
   \,=\, \partial_\mu \omega^{(2)}(\mu_*,u_0)\,\mu_*'(u_0)\,v
         \,+\, \partial_u\omega^{(2)}(\mu_*,u_0)\,v\,,
\ee
where $\mu_*=\mu_*(u_0)$. The partial derivative of the cluster function
with respect to the chemical potential is given by
\be{dmuom2}
\begin{aligned}
   \partial_\mu \omega^{(2)}(\mu,u)
   &\,=\, \partial_\mu\bigl(\pcf(\mu,u)\,-\, \rho(\mu,u)^2\bigr) \\[1ex]
   &\,=\, \partial_\mu \pcf(\mu,u) \,-\, 2\rho(\mu,u)\partial_\mu\rho(\mu,u)\,.
\end{aligned}
\ee
From \req{durho} and \req{nablarho} we conclude that 
\begin{align*}
   \partial_\mu\rho^{(2)}(x) &\,=\, -2\,\nabla_u\rho\thsp(x)\\[1ex]
   &\,=\, 2\beta\, \pcf(x) \,+\, \beta
         \int_{\R^d} \omega^{(3)}(x,0,x') \dx' 
         \,+\, 2\beta\rho\int_{\R^d}\omega^{(2)}(x')\dx'\,,
\end{align*}
and hence, together with Lemma~\ref{Lem:dmurho} it follows that
\be{dmuomega2}
   \partial_\mu \omega^{(2)}(x) 
   \,=\, 2\beta \,\omega^{(2)}(x) \,+\, 
         \beta\int_{\R^d} \omega^{(3)}(x,0,x')\dx' \,.
\ee
In particular, this shows that $\partial_\mu\omega^{(2)}\in L^1(\R^d)$
by virtue of \req{Thm448}.
Together with Proposition~\ref{Prop:mustardiff} and \req{Fprime-omega} 
we thus obtain
\begin{align*}
   &\norm{F'(u_0)\,v}_{L^1(\R^d)}\\[1ex]
   &\qquad
    \,\leq\, \Bigl(
            \norm{\partial_\mu\omega^{(2)}(\mu_*,u_0)}_{L^1(\R^d)}
            \norm{\mu_*'(u_0)}_{\V'}
            \,+\, \norm{\partial_u\omega^{(2)}(\mu_*,u_0)}_{\L(\V,L^1(\R^d))}
            \Bigr)
            \norm{v}_\V\,,
\end{align*}
because it has been proved in \cite{Hank18b} that
$\partial_u\omega^{(2)}\in\L(\V,L^1(\R^d))$. This shows that
$F'(u_0)\in \L(\V,L^1(\R^d))$.
\end{proof}

Since $L^\infty(\R^d)\cap L^1(\R^d)\subset L^2(\R^d)$,
Theorem~\ref{Thm:Fprime1} shows that $F'(u_0)$ is a bounded operator 
from $\V$ to $L^2(\R^d)$.
As $\V$ is a dense subspace of $L^2(\R^d)$, and because we have repeatedly
identified the dual pairing $\V\times\V'$ with the dual pairing
of $L^2(\R^d)$, this raises the question whether $F'(u_0)$ extends to
a bounded operator in $L^2(\R^d)$, or, in terms of the relative entropy
functional, whether $\Phi''(u_0)$ extends to a quadratic form on $L^2(\R^d)$.
This will be answered to the affirmative in the remainder of this section.

\begin{lemma}
\label{Lem:Fprime}
Under the assumptions of Theorem~\ref{Thm:Phiprime} the 
Jacobian of the Henderson map is given by
\begin{subequations}
\label{eq:finalform}
\begin{align}
\label{eq:finalform-a}
   \bigl(F'(u_0)\,v\bigr)(x) &\,=\, -\beta \rho^{(2)}(x)v(x) \\[1ex]
\label{eq:finalform-b}
   &\phantom{\,=\ } 
          \,-\, 2\beta\rho\,(\omega^{(2)}*v)(x)
          \,-\, \beta\,(\omega^{(2)}*\omega^{(2)}*v)(x)\\[1ex]
\label{eq:finalform-c}
   &\phantom{\,=\ }
          \,+\, \frac{1}{2\partial_\mu\rho}\,\partial_\mu\omega^{(2)}(x)
                \int_{\R^d}\partial_\mu\omega^{(2)}(x')\,v(x')\dx'\\[1ex]
\label{eq:finalform-d}
   &\phantom{\,=\ }
          \,-\, 2\beta\int_{\R^d}\omega^{(3)}(x,0,x')\,v(x')\dx'\\[1ex]
\label{eq:finalform-e}
   &\phantom{\,=\ }
          \,-\, \frac{\beta}{2}\int_{\R^d}
                   \int_{\R^d} \omega^{(4)}(x,0,x'',x'+x'')\dx''\,v(x')\dx'
\end{align}
\end{subequations}
for $v\in \V$ and $x\in\R^d$.
\end{lemma}

\begin{proof}
According to \req{A} we have
\[
   F'(u_0)\thsp v \,=\, -2Av 
            \,=\, \partial_u\rho^{(2)}\,v 
                  \,+\, \frac{2}{\partial_\mu\rho}
                        (\partial_u\rho\,v)\nabla_u\rho\,,
\]
where
\be{durho2}
\begin{aligned}
   &(\partial_u\rho^{(2)}\,v)(x) 
    \,=\, -\beta\pcf(x)v(x) \,-\, 2\beta \int_{\R^d} \rho^{(3)}(x,0,x')v(x')\dx'
         \\[1ex]
   &\qquad 
          \,-\, \frac{\beta}{2}\int_{\R^d} v(x')
                \int_{\R^d} \bigl(\rho^{(4)}(x,0,x'',x''+x')
                            \,-\, \rho^{(2)}(x)\rho^{(2)}(x')\bigr)\dx''\dx'
\end{aligned}
\ee
by virtue of \cite[Eq.~(6.9)]{Hank18b}. Concerning the nested integral
in the second line of \req{durho2} it has 
further been shown in \cite[Proposition~4.3]{Hank18b} that the
inner integral converges absolutely
and defines an $L^\infty$ function of $x$ and $x'$.
It follows that
\be{Lem:Fprime0}
   \bigl(F'(u_0)v\bigr)(x) 
   \,=\, -\beta\rho^{(2)}(x)\,v(x) \,+\, \int_{\R^d} k(x,x')v(x')\dx'
\ee
with integral kernel
\be{k}
\begin{aligned}
   k(x,x') 
   &\,=\, \frac{2}{\partial_\mu\rho} \nabla_u\rho(x)\nabla_u\rho(x')
        \,-\,2\beta\,\rho^{(3)}(x,0,x')\\[1ex]
   &\phantom{\,=\, }
        \,-\, \frac{\beta}{2}
        \int_{\R^d} \bigl(\rho^{(4)}(x,0,x'',x''+x')
                            \,-\, \rho^{(2)}(x)\rho^{(2)}(x')\bigr)\dx''\,.
\end{aligned}
\ee

We now rewrite the individual terms of this kernel function.
To begin with, we recall from \req{durho} and \req{dmuom2} that
\[
   \nabla_u\rho \,=\, -\frac{1}{2}\,\partial_\mu\rho^{(2)}
   \,=\, -\frac{1}{2}\,\partial_\mu\omega^{(2)} \,-\, \rho\,\partial_\mu\rho\,,
\]
so that
\begin{align*}
   \frac{2}{\partial_\mu\rho}\nabla_u\rho(x)\nabla_u\rho(x')
   &\,=\, \frac{1}{2\partial_\mu\rho}\,\partial_\mu\omega^{(2)}(x)\,
          \partial_\mu\omega^{(2)}(x')\\[1ex]
   &\phantom{\,=\ }
         \,+\, \rho\,\partial_\mu\omega^{(2)}(x) 
         \,+\, \rho\,\partial_\mu\omega^{(2)}(x') 
         \,+\, 2\rho^2\partial_\mu\rho\,.
\end{align*}
The partial derivatives with respect to $\mu$ in the final three terms
of the right-hand side can be eliminated with the help of \req{dmuomega2} 
and Lemma~\ref{Lem:dmurho}:
This yields
\be{Lem:Fprime1}
\begin{aligned}
   &\frac{2}{\partial_\mu\rho}\nabla_u\rho(x)\nabla_u\rho(x')
    \,=\, \frac{1}{2\partial_\mu\rho}\,\partial_\mu\omega^{(2)}(x)\,
          \partial_\mu\omega^{(2)}(x')\\[1ex]
   &\qquad
          \,+\, 2\beta\rho\,\omega^{(2)}(x)
          \,+\, \beta\rho\int_{\R^d}\omega^{(3)}(x,0,x'')\dx''
          \,+\, 2\beta\rho\,\omega^{(2)}(x')\\[1ex]
   &\qquad 
          \,+\, \beta\rho\int_{\R^d}\omega^{(3)}(x',0,x'')\dx'' 
          \,+\, 2\beta\rho^3 
          \,+\, 2\beta\rho^2\int_{\R^d} \omega^{(2)}(x'')\dx''\,.
\end{aligned}
\ee

Concerning the second term of the integral kernel~\req{k} we apply
the definition~\req{omegan} of the cluster functions to rewrite
\be{Lem:Fprime2}
    \rho^{(3)}(x,0,x')
    \,=\, \omega^{(3)}(x,0,x') \,+\, \rho\,\omega^{(2)}(x)
          \,+\, \rho\,\omega^{(2)}(x-x') \,+\, \rho\,\omega^{(2)}(x')+\rho^3\,,
\ee
where we have also used \req{om2even}.

In the same way we can rewrite the integrand of the 
integral in \req{k} in terms of the cluster functions:
\begin{align*}
    &\rho^{(4)}(x,0,x'',x''+x') \,-\, \pcf(x)\pcf(x')  \\
    &\qquad\,=\, \omega^{(4)}(x,0,x'',x''+x') 
     \,+\, \omega^{(3)}(x,0,x'')\, \rho
     \,+\, \omega^{(3)}(x,0,x''+x')\, \rho \\
    &\qquad\phantom{\,=\, }
     \,+\,\omega^{(3)}(x,x'',x''+x')\, \rho
     \,+\, \omega^{(3)}(0,x'',x''+x')\,\rho  \\
    &\qquad\phantom{\,=\, }
     \,+\, \omega^{(2)}(x-x'')\,\omega^{(2)}(x''+x') 
     \,+\, \omega^{(2)}(x-x'-x'')\,\omega^{(2)}(x'') \\
    &\qquad\phantom{\,=\, }
     \,+\, \omega^{(2)}(x''-x)\,\rho^2 
     \,+\, \omega^{(2)}(x''+x'-x)\,\rho^2  \\
    &\qquad\phantom{\,=\, }
     \,+\, \omega^{(2)}(x'')\,\rho^2 
     \,+\, \omega^{(2)}(x''+x')\,\rho^2 \\
    &\qquad\,=\, \omega^{(4)}(x,0,x'',x''+x') 
     \,+\, \omega^{(3)}(x,0,x'')\, \rho
     \,+\, \omega^{(3)}(x,0,x''+x')\, \rho \\
    &\qquad\phantom{\,=\, }
     \,+\,\omega^{(3)}(x',0,x-x'')\, \rho
     \,+\, \omega^{(3)}(x',0,-x'')\,\rho  \\
    &\qquad\phantom{\,=\, }
     \,+\, \omega^{(2)}(x-x'')\,\omega^{(2)}(x''+x') 
     \,+\, \omega^{(2)}(x-x'-x'')\,\omega^{(2)}(x'') \\
    &\qquad\phantom{\,=\, }
     \,+\, \omega^{(2)}(x''-x)\,\rho^2 
     \,+\, \omega^{(2)}(x''+x'-x)\,\rho^2  \\
    &\qquad\phantom{\,=\, }
     \,+\, \omega^{(2)}(x'')\,\rho^2 
     \,+\, \omega^{(2)}(x''+x')\,\rho^2 ,
\end{align*}
where we have used in the second step that the cluster functions are 
translation and permutation invariant.
Integrating over $x''\in\R^d$ we thus obtain
\be{Lem:Fprime3}
\begin{aligned}
    &\int_{\R^d}
     \bigl(\rho^{(4)}(x,0,x'',x''+x') \,-\, \pcf(x)\pcf(x')\bigr)\dx''\\[1ex]
    &\qquad\qquad
     \,=\, \int_{\R^d} \omega^{(4)}(x,0,x'',x''+x') \dx''
           \,+\, 4\rho^2\int_{\R^d}\omega^{(2)}(x'')\dx''\\[1ex]
    &\qquad\qquad\phantom{\,=\, }
     \,+\, 2\rho\int_{\R^d}\omega^{(3)}(x,0,x'')\dx'' 
     \,+\, 2\rho\int_{\R^d}\omega^{(3)}(x',0,x'')\dx'' \\[1.5ex]
    &\qquad\qquad\phantom{\,=\, }
     \,+\, \bigl(\omega^{(2)}*\omega^{(2)}\bigr)(x+x') 
     \,+\, \bigl(\omega^{(2)}*\omega^{(2)}\bigr)(x-x') \,,
\end{aligned}
\ee
where all the terms on the right-hand side are bounded functions of $x$ and
$x'$: This follows from \req{om3int} and the fact that 
$\omega^{(2)}\in L^1(\R^d)\cap L^\infty(\R^d)$.

Inserting \req{Lem:Fprime1}, \req{Lem:Fprime2}, and \req{Lem:Fprime3} into
\req{k} we conclude that
\begin{align*}
   k(x,x')
   &\,=\, \frac{1}{2\partial_\mu\rho}\,\partial_\mu\omega^{(2)}(x)\,
          \partial_\mu\omega^{(2)}(x')
          \,-\, 2\beta\,\omega^{(3)}(x,0,x') 
          \,-\, 2\beta\rho\,\omega^{(2)}(x-x')\\[1ex]
   &\phantom{\,=\ }
          \,-\,\frac{\beta}{2}\int_{\R^d}\omega^{(4)}(x,0,x'',x''+x')\dx''\\[1ex]
   &\phantom{\,=\ }
          \,-\,\frac{\beta}{2}\,\bigl(\omega^{(2)}*\omega^{(2)}\bigr)(x+x') 
          \,-\,\frac{\beta}{2}\,\bigl(\omega^{(2)}*\omega^{(2)}\bigr)(x-x')\,.
\end{align*}
Note that the final two terms of this integral kernel representation 
define the same convolution operator
\begin{align*}
    v \, \mapsto\, 
    -\frac{\beta}{2}
    \int_{\R^d}\bigl(\omega^{(2)}* \omega^{(2)}\bigr)(\,\cdot\, -x')\,v(x')\dx'\,,
\end{align*}
because every $v \in \V$ is an even function.
The assertion~\req{finalform} thus follows readily from \req{Lem:Fprime0}.
\end{proof}

Now we can prove the main result of this section.

\begin{theorem}
\label{Thm:Fprime2}
Under the assumptions of Theorem~\ref{Thm:Phiprime}
the operator $F'(u_0)$ extends to a selfadjoint negative semidefinite operator
on $L^2(\R^d)$.
\end{theorem}

\begin{proof}
We already know from Theorem~\ref{Thm:Fprime1}
that $F'(u_0)\in\L(\V,L^2(\R^d))$. Since $\V$ is a dense subset of $L^2(\R^d)$, 
it remains to discuss the continuity from $L^2(\R^d)$ to $L^2(\R^d)$
of each of the terms defined in lines \req{finalform-a}-\req{finalform-e} 
of Lemma~\ref{Lem:Fprime}.

For the multiplication operator in \req{finalform-a} this is true because
the pair correlation function $\pcf$ is bounded. The autoconvolution
$\omega^{(2)}*\omega^{(2)}$ of the cluster function $\omega^{(2)}\in L^1(\R^d)$
belongs to $L^1(\R^d)$ again, hence the two convolution operators in 
\req{finalform-b} are bounded operators in $L^2(\R^d)$. 
From \req{dmuomega2} it follows that $\partial_\mu\omega^{(2)}$ 
belongs to $L^1(\R^d)\cap L^\infty(\R^d)$:
the $L^1$ property has been verified in the argument following \req{dmuomega2};
the $L^\infty$ property of the second term in \req{dmuomega2} has been 
established in \req{om3int}. 
This shows that $\partial_\mu\omega^{(2)}\in L^2(\R^d)$, and the continuity
of the operator in \req{finalform-c} is therefore a consequence of the
Cauchy-Schwarz inequality. By virtue of \req{Thm448}
the cluster function 
$\omega^{(3)}(\,\cdot\,,0,\,\cdot\,)$ belongs to
$L^1(\R^d\times\R^d)$. Since it is bounded, it also belongs to
$L^2(\R^d\times\R^d)$. Therefore, \req{finalform-d} defines a compact operator
in $\L(L^2(\R^d))$. Concerning \req{finalform-e} we have already pointed out
after \req{Lem:Fprime3} that the inner integral of $\omega^{(4)}$
is a bounded function of $x$ and $x'$. It also belongs to $L^1(\R^d\times\R^d)$
by virtue of \req{Thm448}, again. The continuity of the
operator defined in \req{finalform-e} therefore follows in the same way as
in the previous case. 

In summary this shows that $F'(u_0)$ extends to an operator in $\L(L^2(\R^d))$.
By virtue of \req{FprimeHessian} this extension is selfadjoint and negative
semidefinite.
\end{proof}

Concerning the Hessian of the relative entropy functional we readily
conclude from \req{FprimeHessian} the following corollary.

\begin{corollary}
\label{Kor:HessianL2}
Under the assumptions of Theorem~\ref{Thm:Phiprime} the Hessian $\Phi''(u_0)$
of the relative entropy functional~\req{Phi} defines a symmetric and 
positive semidefinite bilinear form on $L^2(\R^d)$.
\end{corollary}

As indicated before we leave it as an open problem whether $F'(u_0)$ is
injective, i.e., whether the local quadratic approximation~\req{Phi2} of the
relative entropy functional is strictly convex.

\section{Lennard-Jones type pair potentials}
\label{Sec:LJtype}
In this final section we establish a stronger regularity result for the
derivative of the Henderson map, respectively the Hessian of the
relative entropy functional, which is valid when the potential $u_0\in\U$
satisfies \req{potentials} with a majorant of the form
\be{LJtype}
   \psi_0(r) \,=\, C(1+r^2)^{-\alpha/2}
\ee
for some $C>0$ and $\alpha>d$. This class of potentials includes the so-called
potentials of Lennard-Jones type, cf.~\cite{Ruel69}.

It it known that in the gas phase the cluster function $\omega^{(2)}$
corresponding to such a pair potential satisfies the same rate of decay 
near infinity as $\psi_0$ does. This entails the following result.

\begin{theorem}
\label{Thm:LJtype}
Let $u_0$ satisfy \req{potentials} for some $\varphi$ as in 
\req{varphi,psi} and $\psi=\psi_0$ of \req{LJtype}. Let the space $\V$ of 
perturbations of $u_0$ be defined as before, cf.~\req{Vd}, 
and let $\rho_*$ belong to the gas phase of $u_0$. Then $F'(u_0)\in\L(\V)$.
\end{theorem}

\begin{proof}
The aforementioned result about the cluster function states that
$\omega^{(2)}\in\V$ under the given assumptions;
compare Lemma~\ref{Lem:Frommer} in the appendix.

Now we discuss the continuity of each of the operators corresponding
to the individual terms in \req{finalform}.
Continuity obviously holds for the multiplication operator 
in \req{finalform-a}, because $\pcf$ is bounded. 
Concerning the convolutions in \req{finalform-b} we refer to 
\cite[Proposition~4.1]{Hank18c} for the fact that the convolution is a 
continuous bilinear mapping from $\V\times\V$ to $\V$. Accordingly, the two
terms in \req{finalform-b} also define operators in $\L(\V)$.

For the remaining terms we make use of Lemma~\ref{Lem:Frommer} again. 
For $m=3$ this auxiliary result gives
\be{intomega3}
   \int_{\R^d} \bigl|\omega^{(3)}(\,\cdot\,,0,x')\bigr| \dx' \ \in \ \V,
\ee
and since
\[
   \left|\int_{\R^d} \omega^{(3)}(x,0,x')\,v(x')\dx'\right|
   \,\leq\, \norm{v}_{L^\infty(\R^d)}
            \int_{\R^d} \bigl|\omega^{(3)}(x,0,x')\bigr| \dx'
\]
for every $x\in\R^d$, it follows that
\[
   \left\| \int_{\R^d}\omega^{(3)}(\,\cdot\,,0,x')\,v(x')\dx'\right\|_\V
   \,\leq\, \norm{v}_{L^\infty(\R^d)}
            \left\| \int_{\R^d}
               \bigl|\omega^{(3)}(\,\cdot\,,0,x')\bigr| \dx' \right\|_\V .
\]
Much the same argument applies to the term in \req{finalform-e}:
\begin{align*}
   &\left|\int_{\R^d}\int_{\R^d}     
             \omega^{(4)}(x,0,x'',x'+x'')\dx''\,v(x')\dx'\right|\\[1ex]
   &\qquad
    \,\leq\, \norm{v}_{L^\infty(\R^d)}
             \int_{\R^d}\int_{\R^d} \bigl|\omega^{(4)}(x,0,x'',x')\bigr|
             \dx'\dx''
\end{align*}
for every $x\in\R^d$,
and the right-hand side again belongs to $\V$ as a function of $x$ 
by virtue of Lemma~\ref{Lem:Frommer} for $m=4$. 
Since $\V$ is continuously embedded in $L^\infty(\R^d)$
this shows that the two terms in \req{finalform-d} and \req{finalform-e}
correspond to operators in $\L(\V)$.

Finally, concerning the term in \req{finalform-c} we conclude from
\req{dmuomega2}, \req{intomega3}, and the fact that $\omega^{(2)}\in\V$ 
that $\partial_\mu\omega^{(2)}$ belongs to $\V$ as well.
Accordingly, this term also represents a bounded operator from $\V$ to $\V$,
and the proof is done.
\end{proof}

Note that if the target $\P_*$ and the model $\P_0$ are Gibbs measures
corresponding to pair potentials $u_*$ and $u_0$ in $\U$ and their associated
chemical potentials, where $u_*$ and $u$ satisfy
\req{potentials} for the same majorant $\psi=\psi_0$ given by \req{LJtype},
and if the density $\rho_*$ belongs to the gas phase of both pair potentials,
then the cluster functions $\omega^{(2)}_*$ and $\omega^{(2)}_0$ of
$\P_*$ and $\P_0$, respectively, both belong to $\V$. It thus follows from
Remark~\ref{Rem:nablaPhi} that the gradient $\nabla\Phi(u_0)$ of the
relative entropy functional belongs to $\V$ as well, and therefore
the Newton equation~\req{IMC} of the IMC iterative scheme is an operator
equation in $\V$ by virtue of Theorem~\ref{Thm:LJtype}.
As shown in \cite[Remark~6.5]{Hank18c}, a valid choice $u_0\in\U$, 
which satisfies~\req{potentials} for the same majorant \req{LJtype}, 
is given by the potential of mean force,
\[
   u_0 \,=\, -\frac{1}{\beta}\,\log\bigl(\pcf_*/\rho_*^2\bigr)
   \,=\, -\frac{1}{\beta}\,\log\bigl(1+\omega^{(2)}_*/\rho_*^2\bigr)\,,
\]
provided the density $\rho_*$ is sufficiently small. 
Moreover, $u_*-u_0$ belongs to $\V$ for this initial guess.

\renewcommand{\thesection}{\Alph{section}}
\setcounter{section}{1}
\setcounter{equation}{0}
\setcounter{theorem}{0}
\section*{Appendix}
In this appendix we provide some auxiliary results. In doing so we will 
repeatedly make use of the short-hand notation
\[
   \xx_n \,=\, (x_1,\dots,x_n) \,\in\,(\R^d)^n\,.
\]

The first result is needed for the proof of Theorem~\ref{Thm:Chayes},
where we construct a $(\mu,u)$-Gibbs measure for a pair potential $u\in\U$
with prescribed density $\rho_*<\rcp(u)$.

\begin{lemma}
\label{Lem:gibbscontinuity}
Let $u\in\U$ and assume that the sequence $(\mu_k)_k\subset \R$ satisfies
$\mu_k\to\mu\in\R$ as $k\to\infty$. Furthermore, for every $k\in\N$ let 
$\P_k$ be a $(\mu_k,u)$-Gibbs measure.
Then $(\P_k)_k$ has an accumulation point with respect to the local topology,
and every such accumulation point is a $(\mu,u)$-Gibbs measure. 
\end{lemma}

\begin{proof}
Let $(\rho_k^{(m)})_m$ be the sequence of correlation functions of $\P_k$.
Since the sequence $(\mu_k)_k $ is bounded, 
the measures $\P_k$ satisfy
a uniform Ruelle condition by virtue of \cite[Corollary 5.3]{Ruel70},
i.e., there exists $\xi>0$, such that
\be{Ruellebound}
        \bigl\|\rho_k^{(m)}\bigr\|_{L^\infty((\R^d)^m)} \,\leq\, \xi^m \qquad 
        \text{ for every }\ k,m \in\N\,.
\ee
This implies that for fixed $m\in\N$ the sequence $(\rho_k^{(m)})_k$ of 
correlation functions has a weak$^*$ convergent subsequence in
$L^\infty((\R^d)^m)$ with limit 
$\rho^{(m)}$, say.
With a diagonalisation argument we can thus find a subsequence of indices 
$k\in\N$, for which
\be{weakstar}
    \rho^{(m)}_k \rightharpoonup^* \rho^{(m)}, \qquad k \to \infty\,,
\ee
holds for every $m\in\N$. This is equivalent to saying that
$(\P_k)_k$ has a subsequence which converges to some probability measure $\P$
in the local topology.
To simplify the notation we assume in the sequel that the entire sequence
$(\P_k)_k$ is convergent.

Since every $\P_k$ is a Gibbs measure, the respective correlation functions 
satisfy the Kirkwood-Salsburg equations, 
cf., e.g., \cite[Corollary 5.3]{Ruel70}: 
For each $k\in\N$ and every $m\in\N_0$ there holds
\be{KS2}
\begin{aligned}
        &\rho^{(m+1)}_k(x_0,\xx_m) 
         \,=\, e^{ \beta\mu_k } e^{-\beta W(x_0; \xx_m)}\,\cdot \\[1ex]
        &\qquad\qquad 
         \left(\rho_k^{(m)}(\xx_m)
         \,+\,\sum_{n=1}^\infty \frac{1}{n!}\int_{ (\R^d)^n}
         \prod_{j=1}^n
         f(x_0-y_j)\rho^{(m+n)}_k(\xx_m,\yy_n)\dyy_n\right),
\end{aligned}
\ee
where $\rho_k^{(0)}=1$, $\yy_n=(y_1,\dots,y_n)\in(\R^d)^n$,
\be{Mayer}
   f(x) \,=\,  e^{- \beta u(x)}-1
\ee
is the well-known Mayer-function, and  
\[
   W(x_0;\xx_m) \,=\, \sum_{i=1}^m u(x_0-x_i)\,,
\]
the latter being taken to be zero for $m=0$. 
Note that the right-hand side of 
\req{KS2} converges in $L^\infty((\R^d)^{m+1})$
by virtue of the Ruelle condition~\req{Ruellebound} and the fact that
the Mayer function~\req{Mayer} belongs to $L^1(\R^d)$; compare~\req{cbeta}.
Moreover, since the Ruelle condition is uniform
in $k$, this convergence is also uniform with respect to $k$, and hence,
the weak$^*$ convergence~\req{weakstar} and the convergence $\mu_k\to\mu$
imply that the Kirkwood-Salsburg equations~\req{KS2}
also hold true for the limiting correlation functions $(\rho^{(m)})_m$.
It thus follows from \cite[Corollary 5.3]{Ruel70} again that $\P$ is a 
tempered $(\mu,u)$-Gibbs measure.
Finally, since $\rho_k^{(m)}$ is translation invariant for every $k\in\N$
and every $m\in\N$,
its weak$^*$ limits $\rho^{(m)}$ must also be translation invariant, 
and hence $\P$ is translation invariant.
\end{proof}

The remaining two results of the appendix concern properties of the
cluster functions (see Remark~\ref{Rem:Ursell}) within the gas phase.
The proofs make use of their \emph{cluster expansions}
(compare \cite[Section~4.4]{Ruel69} and \cite{Stel64})
\be{OmegaCluster}
    \omega^{(m)}(\xx_m) \,=\, 
    \sum_{n=m}^\infty \frac{e^{n\beta\mu}}{(n-m)!}
    \int_{(\R^d)^{n-m}}\phi^{(n)}(\xx_n) \dx_{m+1}\cdots\dx_n
\ee
for $\xx_m\in(\R^d)^m$ with
\be{cluster}
    \phi^{(n)}(\xx_n) \,=\, 
    \sum_{{\cal C}\in\mathfrak{C}_{n}} \prod_{(i,j)\in {\cal C}} f(x_i-x_j) \,,
\ee
where $f$ is, again, the Mayer function~\req{Mayer}, 
and $\mathfrak{C}_n$ denotes the set of undirected connected graphs  
with vertices $\{1,\dots,n\}$ and edges $(i,j)\in{\cal C}$, which connect
vertices $i$ and $j$;
for $n=m$ the integral in \req{OmegaCluster} has to be replaced by 
the value $\phi^{(m)}(\xx_m)$ of its integrand. 
These cluster expansions are absolutely convergent
for $u_0\in\U$ and $\mu\leq\mu_0$ with $\mu_0$ of \req{mu0}.

%

\begin{lemma}
\label{Lem:dmurho}
Let $u_0\in\U$ and $\mu<\mu_0$ with $\mu_0$ as in \req{mu0}.
Then the density $\rho$ of the corresponding Gibbs measure
is differentiable with respect to $\mu$ and the derivative is given by
\be{dmurho}
   \partial_\mu \rho \,=\, 
   \beta \rho \,+\, \beta \int_{\R^d}\omega^{(2)}(x)\dx\,.
\ee
\end{lemma}

\begin{proof}
By \req{OmegaCluster} and \req{omegan} there holds (with $x_1=0$)
\begin{align*}
    \rho \,=\, \omega^{(1)}(0)
    \,=\, \sum_{n=1}^\infty \frac{e^{n\beta\mu}}{(n-1)!} 
          \sum_{{\cal C}\in\mathfrak{C}_n} 
          \int_{(\R^d)^{n-1}} \prod_{(i,j)\in \cal{C}} f(x_i-x_j) 
          \dx_2\cdots\dx_n\,.
\end{align*}
Differentiating with respect to $\mu $ we obtain
\begin{align*}
   \partial_\mu \rho 
   &\,=\, \sum_{n=1}^\infty \frac{n\beta\,e^{n\beta\mu}}{(n-1)!}              
          \sum_{{\cal C}\in\mathfrak{C}_n} 
          \int_{(\R^d)^{n-1}} \prod_{(i,j)\in \cal{C}} f(x_i-x_j)
          \dx_2\cdots\dx_n\\[1ex]
   &\,=\, \beta \sum_{n=1}^\infty \frac{e^{n\beta\mu}}{(n-1)!}           
          \sum_{{\cal C}\in\mathfrak{C}_n} 
          \int_{(\R^d)^{n-1}} \prod_{(i,j)\in \cal{C}} f(x_i-x_j)
          \dx_2\cdots\dx_n\\[1ex]
   &\phantom{\,=\ }
          \,+\, \beta\int_{\R^d}\left(
                 \sum_{n=2}^\infty \frac{e^{n\beta\mu}}{(n-2)!}           
                 \sum_{{\cal C}\in\mathfrak{C}_n} 
                 \int_{(\R^d)^{n-2}} 
                    \prod_{(i,j)\in \cal{C}} f(x_i-x_j) 
          \dx_3\cdots\dx_n\right)\!\dx_2\\[1ex]
   &\,=\, \beta\rho \,+\, \beta\int_{\R^d}\omega^{(2)}(0,x_2)\dx_2\,.
\end{align*}
Using the short-hand notation for $\omega^{(2)}$ and \req{om2even},
the assertion follows.
\end{proof}
\begin{remark}
\rm
Lemma~\ref{Lem:dmurho} is general folklore, at least its 
``grand canonical version''
\be{dmurho-grandcanonical}
   \partial_\mu \rho_\Lambda(0)
   \,=\, \beta \rho_\Lambda(0) \,+\, 
         \beta \int_\Lambda 
             \Bigl(\pcf_\Lambda(0,x) - 
             \rho_\Lambda(0)\rho_\Lambda(x)\Bigr)\dx
\ee
for a bounded domain $\Lambda\subset\R^d$ is well-known, 
where $\rho_\Lambda$ and $\pcf_\Lambda$ denote the corresponding 
grand canonical correlation functions;
compare, e.g., \cite[Eq.~(2-9)]{Stel64}, or Ben-Naim~\cite[Eq.~(1.53)]{BenN06}.
An alternative proof of \req{dmurho} is thus obtained by passing in
\req{dmurho-grandcanonical} to the thermodynamic limit.
\fin
\end{remark}

Our final lemma strenghtens Ruelle's result~\req{Thm448} 
for Lennard-Jones type pair potentials, and extends 
\cite[Corollary~5.2]{Hank18c} to higher order 
cluster functions.\footnote{
The results in \cite{Hank18c}, like those in \cite{Hank18a,Hank18b}, extend to
the present setting.}

\begin{lemma}
\label{Lem:Frommer}
    Let $u \in\U$ be a pair potential which satisfies \req{varphi,psi},
    \req{potentials} with majorant $\psi=\psi_0$ as in \req{LJtype},
    and let $\omega^{(m)}$, $m\in\N$, be the associated cluster functions.
    Further, let the chemical potential satisfy $\mu\leq\mu_0$, cf.~\req{mu0}.
    Then $\omega^{(2)}\in\V$, and for $m\geq 3$ there holds
    \begin{align*}
        \int_{(\R^d)^{m-2}} 
           \bigl|\omega^{(m)}(\,\cdot\,,0,x_3,\dots,x_m)\bigr| \dx_{2}\cdots
        \dx_m  \in \V.
    \end{align*}
\end{lemma}

\begin{proof}
For $\mu<\mu_0$ it follows from \cite[Lemma 5.1]{Hank18c} that
$\phi^{(n)}$ of \req{cluster} satisfies
\begin{align*}
   \int_{(\R^d)^{n-2}}\bigl| \phi^{(n)}(\xx_n)\bigr| \dx_3\cdots\dx_n
   &\,\leq\, C_\Sigma\, \bigl(c_\beta e^{\beta B}\bigr)^n 
             n^{n-2}\,\psi_0(|x_1-x_2|)
\end{align*}
for some constant $C_\Sigma>0$,  
every $x_1,x_2\in\R^d$ and every $n\geq 2$,
where $c_\beta$ and $B$ are the constants defined in \req{cbeta}, \req{B};
for $n=2$ the left-hand side of this inequality has to be replaced by
\[
   \bigl| \phi^{(2)}(x_1,x_2)\bigr| \,=\, \bigl|f(x_1-x_2)\bigr|\,.
\]
Plugging this into the cluster expansion~\req{OmegaCluster} for $m\geq 3$
it follows that
\begin{align*}
   \int_{(\R^d)^{m-2}}
       \bigl|\omega^{(m)}(\xx_m)\bigr| \dx_{3}\cdots \dx_m 
   &\,\leq\, C_\Sigma\,\psi_0(|x_1-x_2|)
             \sum_{n=m}^\infty
             \frac{n^{n-2}}{(n-m)!}\,
             \bigl(c_\beta e^{\beta (B+\mu)}\bigr)^n\\[1ex]
   &\,\leq\, C_\Sigma\,\psi_0(|x_1-x_2|)
             \sum_{n=m}^\infty
             n^{m-2}
             \bigl(c_\beta e^{\beta (B+\mu)+1}\bigr)^n\,,
\end{align*}
where we have used the inequality $n^{n-m}/(n-m)!\leq e^n$ in the final step.
Since $\mu\leq\mu_0$, the positive number $q=c_\beta e^{\beta(B+\mu)+1}$ is 
below one, and hence the series on the right-hand side converges, 
its value being $C_{\mu,m}$, say. 
This estimate extends to $m=2$ when the left-hand side is replaced by
$\bigl|\omega^{(2)}(\xx_2)\bigr|$.
The assertion therefore follows from \req{Vd}.
\end{proof}

We refer to \cite{Dorl19,Dun75,Min77} for further results on the
decay of the higher order cluster functions.



\begin{thebibliography}{99}
\bibitem{BenN06}
{\sc Ben-Naim, A.}:
   Molecular Theory of Solutions.
   Oxford University Press, New York (2006)
\bibitem{JTC83}
{\sc Chayes, J.T., Chayes, L.}:
   On the validity of the inverse conjecture in classical 
   density functional theory.
   J. Stat. Phys.~\textbf{36}, 471--488 (1984)
\bibitem{Dorl19}
{\sc Dorlas, T.C., Rebenko, A.L., Savoie, B.}:
   Correlation of clusters: Partially truncated correlation functions
   and their decay.
   J. Math. Phys.~{\bf 61}, 033303 (2020)
\bibitem{Dun75}
{\sc Duneau, M., Iagolnitzer, D., Souillard, B.}:
   Decay of correlations for infinite-range interactions.
   J. Math. Phys.~{\bf 16}, 1662--1666 (1975)
\bibitem{ErAd94}
{\sc Ercolessi, F., Adams, J.B.}:
   Interatomic potentials from first-principles calculations: 
   the force-matching method.
   Europhys. Lett.~{\bf 26}, 583--588 (1994)
\bibitem{FJK12}
{\sc Fritsch, S., Junghans, C., Kremer, K.}:
   Structure formation of {T}oluene around {C}60: Implementation of the
   Adaptive Resolution Scheme ({A}d{R}es{S}) into {G}{R}{O}{M}{A}{C}{S}.
   J. Chem. Theor. Comput.~{\bf 8}, 398--403 (2012)
\bibitem{FHJ19}
{\sc Frommer, F., Hanke, M., Jansen, S.}:
   A note on the uniqueness result for the inverse Henderson problem.
   J. Math. Phys.~{\bf 60}, 093303 (2019)
\bibitem{Gall68}
{\sc Gallavotti, G. Miracle-Sole, S.}:
   A variational principle for the equilibrium of hard sphere systems.
   Ann. IHP, Phys. théor.~\textbf{8}, 287--299 (1968)
\bibitem{HGEO94}
{\sc Georgii, H.-O.}:
   Large deviations and the equivalence of ensembles for
   {G}ibbsian particle systems with superstable interaction.
   Probab. Theory Related Fields~\textbf{99}, 171--195 (1994)
\bibitem{HGEO95}
{\sc Georgii, H.-O.}:
   The equivalence of ensembles for classical systems of particles.
   J. Stat. Phys.~\textbf{80}, 1341--1378 (1995)
\bibitem{Geo11}
{\sc Georgii, H.-O.}:
   Gibbs Measures and Phase Transitions.
   de Gruyter, Berlin, New York (2011) 
\bibitem{Geo13}
{\sc Georgii, H.-O.}:
   Stochastics: Introduction to Probability and Statistics, 2nd edn.
   de Gruyter, Berlin, Boston (2013) 
\bibitem{HGEO93}
{\sc Georgii, H.-O., Zessin, H.}: 
   Large deviations and the maximum entropy principle for marked point 
   random fields. 
   Probab. Theory Related Fields~\textbf{96}, 177--204 (1993)
\bibitem{Hank18a}
{\sc Hanke, M.}:
   Fr{\'e}chet differentiability of molecular distribution functions I:
   $L^\infty$ analysis.
   Lett. Math. Phys.~\textbf{108}, 285--306 (2018)
\bibitem{Hank18b}
{\sc Hanke, M.}:
   Fr{\'e}chet differentiability of molecular distribution functions II: 
   the Ursell function.
   Lett. Math. Phys.~\textbf{108}, 307--329 (2018)
\bibitem{Hank18c}
{\sc Hanke, M.}:
   Well-Posedness of the iterative Boltzmann inversion. 
   J. Stat. Phys.~{\bf 170}, 536--553 (2018)
\bibitem{HaMcD13}
{\sc Hansen, J.-P., McDonald, I.R.}:
   Theory of Simple Liquids: With Applications to Soft Matter, 4th edn.
   Academic Press, Oxford (2013)
\bibitem{Hend74}
{\sc Henderson, R.L.}:
   A uniqueness theorem for fluid pair correlation functions.
   Phys. Lett. A~{\bf 49}, 197--198 (1974)
\bibitem{Hugh21}
{\sc Hughes, A.M.}:
   Entropy minimization, convergence and Gibbs ensembles (local and global).
   PhD-Thesis, University of Missouri (2021)
\bibitem{IzVo05}
{\sc Izvekov, S., Voth, G.A.}:
   Multiscale coarse graining of liquid-state systems.
   J. Chem. Phys.~{\bf 123}, 134105 (2005)
\bibitem{Jansen19}
{\sc Jansen, S.}:
   Cluster expansions for Gibbs point processes.
   Adv. Appl. Probab.~{\bf 51}, 1129--1178 (2019)
\bibitem{JKT19}
{\sc Jansen, S., Kuna, T., Tsagkarogiannis, D.}:
   Virial inversion and density functionals,
   arXiv:1906.02322 [math-ph] (2019)
\bibitem{Koral07}
{\sc Koralov, L.}:
   An inverse problem for Gibbs fields with hard core potential.
   J. Math. Phys.~{\bf 48}, 053301 (2007)
\bibitem{Kuna7}
{\sc Kuna, T., Lebowitz, J., Speer, E.}:
   Realizability of point processes.
   J. Stat. Phys.~\textbf{129}, 417--439 (2007)
\bibitem{LLV10}
{\sc Larini, L., Lu, L., Voth, G.A.}:
  The multiscale coare-graining method. VI. Implementation of three-body
  coarse-grained potentials.
  J. Chem. Phys.~{\bf 132}, 164107 (2010)
\bibitem{LyLa95}
{\sc Lyubartsev, A.P., Laaksonen, A.}:
   Calculation of effective interaction potentials from radial distribution
   functions: A reverse Monte Carlo approach.
   Phys. Rev. E~{\bf 52}, 3730--3737 (1995)
\bibitem{Min77}
{\sc Minlos, R.A., Poghosyan, S.}:
   Estimates of Ursell functions, group functions and their derivatives.
   Theoret. and Math. Phys.~{\bf 31}, 408--418 (1977)
\bibitem{MoSa13}
{\sc Monticelli, L., Salonen, E.} (eds.):
   Biomolecular Simulations. Methods and Protocols.
   Springer, New York (2013)
\bibitem{MKV09}
{\sc Murtola, T., Karttunen, M., Vattulainen, I.}:
   Systematic coarse graining from structure using internal states: 
   Application to phospholipid/cholesterol bilayer.
   J. Chem. Phys.~{\bf 131}, 055101 (2009)
\bibitem{Navr16}
{\sc Navrotskaya, I.}:
   Inverse problem in classical statistical mechanics.
   PhD thesis, University of Pitsburgh (2016)
\bibitem{Noid13a}
{\sc Noid, W.G.}:
   Perspective: Coarse-grained models for biomolecular systems.
   J Chem. Phys.~{\bf 139}, 090901 (2013)
\bibitem{Noid13b}
{\sc Noid, W.G.}:
   Systematic methods for structurally consistent coarse-grained models.
   In \cite{MoSa13}, pp.~487--531 (2013)
\bibitem{PeKr09}
{\sc Peter, C., Kremer, K.}:
   Multiscale simulation of soft matter systems -- from the atomistic
   to the coarse-grained level and back.
   Soft Matter~{\bf 5}, 4357--4366 (2009) 
\bibitem{PDK05}
{\sc Praprotnik, M., Delle Site, L., Kremer, K.}:
   Adaptive resolution molecular-dynamics simulation: 
   changing the degrees of freedom on the fly.
   J. Chem. Phys.~\textbf{123}, 224106 (2005)
\bibitem{RPM03}
{\sc Reith, D., P\"utz, M., M\"uller-Plathe, F.}:
   Deriving effective mesoscale potentials from atomistic simulations.
   J. Comput. Chem.~{\bf 24}, 1624--1636 (2003)
\bibitem{RSSV19}
{\sc Rosenberger, D., Sanyal, T., Shell, M.S., van der Vegt, N.F.A.}:
   Transferability of local density-assisted implicit solvation models for
   homogeneous fluid mixtures.
   J. Chem. Theor. Comput.~{\bf 15}, 2881--2895 (2019)
\bibitem{RJLKA09}
{\sc R\"uhle, V., Junghans, C., Lukyanov, A., Kremer, K., Andrienko, D.}:
   Versatile object-oriented toolkit for coarse-graining applications.
   J. Chem. Theor. Comput.~{\bf 5}, 3211--3223 (2009)
\bibitem{Ruel69}
{\sc Ruelle, D.}:
   Statistical Mechanics: Rigorous Results.
   W.A.~Benjamin Publ., New York (1969)
\bibitem{Ruel70}
{\sc Ruelle, D.}:
   Superstable interactions in classical statistical mechanics.
   Comm. Math. Phys.~\textbf{18}, 127--159 (1970)
\bibitem{Schm22}
{\sc Schmid, F.}:
   Editorial: Multiscale simulation methods for soft matter systems.
   J. Phys.: Condens. Matter~{\bf 34}, 160401 (2022)
\bibitem{Shel08}
{\sc Shell, M.S.}:
   The relative entropy is fundamental to multiscale and inverse thermodynamic
   problems.
   J. Chem. Phys.~{\bf 129}, 144108 (2008)
\bibitem{Sope96}
{\sc Soper, A.K.}:
   Empirical potential Monte Carlo simulation of fluid structure.
   Chemical Physics~{\bf 202}, 295--306 (1996)
\bibitem{Stel64}
{\sc Stell, G.}:
   Cluster expansions for classical systems in equilibrium.
   In: Fritsch, H.L., Lebowitz, J.L. (eds.):
   The Equilibrium Theory of Classical Fluids, pp.~II-171--II-266.
   W.A. Benjamin Publ., New York (1964)
\bibitem{Toth07}
{\sc T\'{o}th, G.}:
   Interactions from diffraction data: historical and comprehensive overview
   of simulation assisted methods.
   J. Phys: Condens. Matter~{\bf 19}, 335220 (2007)
\bibitem{THT17}
{\sc Tsourtis, A., Harmandaris, V., Tsagkarogiannis, D.}:
   Parameterization of coarse-grained molecular interactions through
   potential of mean force calculations and cluster expansion techniques.
   Entropy~{\bf 19}, 395 (2017)
\bibitem{WNLV09}
  {\sc Wang, Y., Noid, W.G., Liu, P., Voth, G.A.}:
  Effective force coarse-graining.
  Phys. Chem. Chem. Phys.~{\bf 11}, 2002--2015 (2009)
\end{thebibliography}
\end{document}